\documentclass[11pt,reqno]{amsart}
\allowdisplaybreaks[4]
\usepackage{graphicx} 
\usepackage{subfigure}
\usepackage{tcolorbox}
\usepackage{epsfig}
\usepackage{amssymb}
\usepackage{amsmath,xy}
\usepackage{cite,color}
\newtheorem{theorem}{Theorem}

\newtheorem{proposition}{Proposition}

\newtheorem{corollary}{Corollary}

\newtheorem{lemma}{Lemma}
\newtheorem{remark}{Remark}
\newtheorem{example}{Example}
\newcommand{\be}{\begin{equation}}
\newcommand{\ee}{\end{equation}}
\newcommand{\bea}{\begin{eqnarray}}
\newcommand{\eea}{\end{eqnarray}}
\newcommand{\ba}{\begin{array}}
\newcommand{\ea}{\end{array}}
\newcommand{\bean}{\begin{eqnarray*}}
\newcommand{\eean}{\end{eqnarray*}}

\newcommand{\pa}{\partial}

\begin{document}

\title{The generalized Wronskian solutions of the constrained mKP hierarchy}
\author{Jiayi Zhang$^1$, Jipeng Cheng$^{1,2*}$, Wenchuang Guan$^1$}
\dedicatory {$^{1}$ School of Mathematics, China University of
Mining and Technology, \ Xuzhou, Jiangsu 221116, China,\\
$^{2}$ Jiangsu Center for Applied Mathematics (CUMT), \ Xuzhou, Jiangsu 221116, China.
}

\thanks{*Corresponding author. Email: chengjp@cumt.edu.cn;\ chengjipeng1983@163.com.}

\begin{abstract}
In this paper, we investigate the $(k, m)$-constrained 1st modified Kadomtsev-Petviashvili (mKP) hierarchy $(L^k)_{\leqslant 0}= \sum_{i=1}^m q_i \pa^{-1} r_i \pa$. Here, we obtain the corresponding solutions in the form of generalized Wronskians, which include the Wronskians and Grammians as special cases. Most importantly, these generalized  Wronskian solutions are proved to satisfy the bilinear equations of the $(k, m)$-constrained mKP hierarchy, which is generally nontrivial. Our results here will be helpful in the derivation of the more general addition formulae and polynomial solutions for the 1st mKP hierarchy.
\end{abstract}

\maketitle

\section{Introduction}
\subsection{The 1st mKP hierarchy}
The Kadomtsev-Petviashvili (KP) hierarchy \cite{Mulase1994,Kac2023,vanMoerbeke1994,DJKM,Dickey2003,Harnad2021} has achieved great success in theoretical and mathematical physics. As one important generalization of the KP hierarchy, the $(l-l')$-th modified KP (mKP) hierarchy is given by the following bilinear equation \cite{Jimbo,Harnad2021,Kac2018,Dickey1999}:
 \begin{align*}
{\rm Res}_z z^{l-l'}\tau_l(t-[z^{-1}])\tau_{l'}(t'+[z^{-1}])
e^{\xi(t-t',z)}=0,\quad l\geq l',
 \end{align*}
where ${\rm Res}_z\sum_ia_iz^i=a_{-1}$, $t=(t_1=x,t_2,\ldots)$,  $\xi(t,z)=\sum_{i\geq 1}t_iz^i$ and $[z^{-1}]=(z^{-1}, z^{-2}/2,z^{-3}/3,\ldots)$.\ Note that the 0-th mKP hierarchy is the usual KP hierarchy\cite{DJKM,Dickey2003}, so all $\tau_l(t)$ are KP tau functions. The $(l-l')$-th mKP hierarchy is also known as the discrete KP hierarchy \cite{Adler1999,Dickey1999,Cao2026}, which can be used to describe the KP Darboux orbits \cite{Yang2022,Willox2004}. Moreover, the multi-component generalizations of the mKP hierarchy can be found in \cite{Zabrodin2019,Cao2026}.

Here in this paper, we will be more interested in the 1st mKP hierarchy, which is the simplest mKP hierarchy, satisfying the bilinear equation as follows \cite{Kac2018,cheng2018kp}:
\begin{align}\label{mkpbilnear1}
{\rm Res}_z z\tau_1(t-[z^{-1}])\tau_{0}(t'+[z^{-1}])
e^{\xi(t-t',z)}=0.
\end{align}
And it can be found that \eqref{mkpbilnear1} is equivalent to the following bilinear equation \cite{cheng2018kp,Kac2018}:
\begin{align}\label{mkpbilnear2}
{\rm Res}_z z^{-1}\tau_0(t-[z^{-1}])\tau_1(t'+[z^{-1}])
e^{\xi(t-t',z)}=\tau_0(t')\tau_1(t).
\end{align}
The 1st modified KP hierarchy describes the relation between the KP and the transformed KP hierarchies after a 1-step Darboux transformation\cite{Yang2022,Zabrodin2026,Kac2018}, which can be used to describe the generating functions of open and closed intersection numbers\cite{Alexandrov2015}. In what follows, when we refer to the mKP hierarchy, it usually means the 1st mKP hierarchy. If $\tau_0$ and $\tau_1$ satisfy \eqref{mkpbilnear1} or \eqref{mkpbilnear2}, then $(\tau_0,\tau_1)$ is called the mKP tau pair.

If define the mKP wave function $w(t,z)$ and the mKP adjoint wave function $w^*(t,z)$ as follows \cite{cheng2018kp,Zabrodin2026},
\begin{align}\label{adwavefunction}
w(t,z)=\frac{\tau_0(t-[z^{-1}])}{\tau_1(t)}e^{\xi(t,z)},\  w^*(t,z)=\frac{\tau_1(t+[z^{-1}])}{z\tau_0(t)}e^{-\xi(t,z)},
\end{align}
then \eqref{mkpbilnear2} can be rewritten as
\begin{align*}
{\rm Res}_{z}w(t,z)w^*(t',z)=1.
\end{align*}
Further if introduce the dressing operator $W=e^\alpha+\sum_{i=1}^{+\infty}\beta_i\pa^{-i}$ with $\pa=\pa_x$, such that $w(t,z)=W\left(e^{\xi(t,z)}\right)$, then we can find
$w^*(t,z)=(W^{-1}\pa^{-1})^*\left(e^{-\xi(t,z)}\right)$
and $ W_{t_n}=-\Big(W\pa^nW^{-1}\Big)_{\leq 0}W$. Here for $A=\sum_ia_i\pa^i$, we denote $ A_{\geq j}=\sum_{i\geq j} a_i\pa^i$, $A_{<j}=\sum_{i<j}a_i\pa^i$, $A_{[j]}=a_j$, and the symbol $*$ denotes the adjoint operation: $(\sum_ia_i\pa^i)^*=\sum_i(-1)^i\pa^ia_i$. So if introduce the mKP Lax operator  $L=W\pa W^{-1}$, then $L$ has the following form $L=\pa+\sum_{i=0}^{+\infty} v_i\pa^{-i}$
satisfying the following Lax equation
\begin{align}\label{Lax equation}
L_{t_n}=[(L^n)_{\geq 1},L].
\end{align}
Notice that $(L^n)_{\geq 1}$ is firstly used by Kupershmidt and Kiso \cite{kiso1990,kupershmidt1985} in the construction of the evolution equations, thus the mKP hierarchy \eqref{Lax equation} is also called the Kupershmidt-Kiso mKP hierarchy\cite{cheng2018,cheng2018kp}.
\subsection{The constrained mKP hierarchy}
Here we will focus on an important reduction of the mKP hierarchy: the $(k,m)$-constrained mKP hierarchy, defined by \cite{Kun1995,Oevel1999JP,cheng2018kp}:
\begin{align}
&(L^k)_{\leq 0}=\sum_{i=1}^{m}q_i\pa^{-1}r_i\pa,\quad L_{t_n}=[(L^n)_{\geq 1},L],\label{mKPLK<1}\\
&q_{j,t_n}=(L^n)_{\geq1}(q_j),\quad r_{j,t_n}=-\big(\pa^{-1}(L^n)^*_{\geq1}\pa\big)(r_j),\quad 1\leq j\leq m. \label{mKPeigennew}
\end{align}
This system contains the Chen-Lee-Liu equation, the Gerdjikov-Ivanov equation, the derivative NLS equation, and so on (see\cite{Kun1995}). There exist many important results (e.g.\cite{Chen20193,Chen2020,Kun1995,Oevel1999JP,wang2026,wu2022MMMA,Liu1995}) on the constrained mKP hierarchy. In particular, the $(k,m)$-constrained mKP hierarchy \eqref{mKPLK<1} and \eqref{mKPeigennew} can be expressed by the following bilinear equations in terms of the mKP wave functions\cite{wu2022MMMA,Chen2020}
\begin{align}
&{\rm Res}_zz^kw(t,z)w^*(t',z)=\sum_{i=1}^{m}q_i(t)r_i(t'),\\
&{\rm Res}_zw(t',z){\Omega}(q(t)_x, w^*(t,z))=q_i(t)-q_i(t'),\label{resequation1}\\
&{\rm Res}_zw^*(t',z) \Omega(w(t,z), r(t)_x)=r_i(t)-r_i(t'),\quad 1\leq i\leq m,\label{resequation2}
\end{align}
where  $\Omega(f, g)$ is the squared eigenfunction potential \cite{cheng2018kp,Oevel1998JP} defined by 
\begin{align}\label{omegafg}
\Omega(f, g)_{t_n}= \mathrm{Res}_\pa\left(\pa^{-1}g A_n f \pa^{-1}\right).
\end{align}
Here $f_{t_n}=A_n(f),\ g_{t_n}=A_n^*(g)$ and $A_n$ is a differential operator satisfying
$A_{n,t_m}-A_{m,t_n}-[A_m,A_n]=0$. And the integral constants of $\Omega(f, g)$ in \eqref{resequation1} and \eqref{resequation2} are fixed by the way below:
\begin{align}
&\Omega(q(t)_x,w^*(t,z))=-\Big(\pa^{-1}q(t)_x\pa^{-1}W^{-1*}	\Big)(e^{-\xi(t,z)}),\label{omeq}\\
&\Omega(w(t,z), r(t)_x)=\Big(\pa^{-1}r(t)_xW	\Big)(e^{\xi(t,z)}).\label{omer}
\end{align}
We firstly compute the pseudo-differential operators before $e^{\pm\xi(t,z)}$ in \eqref{omeq}\eqref{omer} and use the fact $\partial^{l}(e^{\pm\xi(t,z)})=(\pm z)^le^{\pm\xi(t,z)}$ to fix the integral constants of $\Omega(f, g)$ in \eqref{resequation1} and \eqref{resequation2}. If further introduce
\begin{align}\label{tau2i-1i}
\tau_{2,i}(t)= q_i(t)\tau_1(t),\  \tau_{-1,i}(t)= r_i(t)\tau_0(t),\quad 1\leq i\leq m,
\end{align}
the $(k,m)$-constrained mKP hierarchy is equivalent to the following bilinear equations in terms of tau functions \cite{Chen2020,wu2022MMMA,wang2026}
\begin{align}
&\mathrm{Res}_z  z^{k-1}
\tau_0\left(t - [z^{-1}]\right)\tau_1\left(t' + [z^{-1}]\right)e^{\xi(t-t',z)}=\sum_{i=1}^{m}\tau_{2,i}(t)\tau_{-1,i}(t'), \label{kmcmkp1}\\
&\mathrm{Res}_z \left( z^{-1}\tau_{a-2,i}(t - [z^{-1}])\tau_{a,i}(t' + [z^{-1}])e^{\xi(t-t',z)} \right)=\tau_{a,i}(t)\tau_{a-2,i}(t'), \ a = 1,2,\label{kmcmkp2}\\
&\mathrm{Res}_z \left( z^{-1}\tau_0(t - [z^{-1}])\tau_1(t'+ [z^{-1}])e^{\xi(t-t',z)} \right)=\tau_1(t)\tau_0(t').\label{kmcmkp3}
\end{align}

In \cite{Chen2020}, the boson-fermion correspondence is used to obtain a set of solutions of the $(k,m)$-constrained mKP hierarchy \eqref{kmcmkp1}-\eqref{kmcmkp3}.
However, these solutions depend on the special choice of group elements, that is
$g=\exp\Big(\sum_{i,j\in \mathbb{Z}} a_{i,j}\psi_i^+\psi_j^-\Big) $ satisfying the condition below
\begin{align}\label{condition}
g^{-1}\Gamma_k g = \sum_{i,j\in\mathbb{Z}} f_{i,j}\psi_i^+\psi_j^-,\quad \Gamma_k = \sum_{i\in\mathbb{Z}} \psi_i^+ \psi_{i+k}^-,
\end{align}
where $f_{i,j} = \sum_{l=1}^m d_i^{(l)} e_j^{(l)}$ for $i \geq 0$, $j \leq 0$, and  $d_i^{(l)}, e_j^{(l)}$ are some constants. And $\psi_i^\lambda$ $(i\in\mathbb{Z}, \lambda=\pm)$ are the charged free fermions satisfying $\psi_i^\lambda \psi_j^\mu+\psi_j^\lambda \psi_i^\mu=\delta_{\lambda+\mu,0}\delta_{i,j}\ (\lambda,\mu=\pm)$. Notice that the condition \eqref{condition} is quite strict and difficult to seek such group elements $g$, thus it is not convenient to construct the corresponding solutions for the constrained mKP hierarchy. In this paper, we will use the mKP Darboux transformations to construct generalized Wronskian solutions of the $(k,m)$-constrained mKP hierarchy, that is, to find the generalized Wronskian form $IW_{N,M}$ \cite{He2002} of $\tau_i(-1\leq i\leq 2)$ to satisfy bilinear equations \eqref{kmcmkp1}-\eqref{kmcmkp3}, defined by:
\begin{equation*}
IW_{N,M}(\psi_{\overrightarrow{N}}; \phi_{\overleftarrow{M}})\triangleq\left|
\begin{array}{ccc}
\Omega(\phi_{1}, \psi_{N})  &\cdots &\Omega(\phi_{M}, \psi_{N})\\
\vdots &\vdots &\vdots\\
\Omega(\phi_{1}, \psi_1) &\cdots &\Omega(\phi_{M}, \psi_{1})\\
\phi_{1}&\cdots&\phi_{M}\\
\vdots &\vdots &\vdots\\
\phi_{1}^{(M-N-1)}&\cdots&\phi_{M}^{(M-N-1)}
\end{array}
\right|, \quad M\geq N.
\end{equation*}
Here $f^{(n)}\triangleq\pa_x^n(f)$, $\overrightarrow{M}\triangleq(M,\ldots,2,1),\overleftarrow{M}\triangleq(1,2,\ldots,M), \psi_{\overrightarrow{N}}\triangleq(\psi_N,\ldots,\psi_1)$ and $\phi_{\overleftarrow{M}}\triangleq(\phi_1,\ldots,\phi_M)$. Note that the generalized Wronskian $IW_{N,M}$ contains the ordinary Wronskians (when $N=0$) and the Grammian (when $M=N$) as special cases. 
\subsection{Main results}
The main results of this paper are given as follows.
\begin{theorem}\label{th1}
Given the mKP tau pair $(\tau_0,\tau_1)$, the mKP eigenfunction $q$ and the mKP adjoint eigenfunction $r$, satisfying $q_{t_n}=(L^n)_{\geq1}(q)$ and $r_{t_n}=-\big(\pa(L^n)_{\geq1}\pa^{-1}\big)^*(r)$, if set
\begin{align*}
  (\tau_0^{[1]},\tau_1^{[1]})=(q \tau_1,-q_x\frac{\tau_1^2}{\tau_0})\quad or \quad (r_x\frac{\tau_0^2}{\tau_1},r\tau_0),
\end{align*}
then $\left(\tau_0^{[1]},\tau_1^{[1]}\right)$ is a new mKP tau pair, satisfying the mKP bilinear equation \eqref{mkpbilnear2}.
\end{theorem}
To express the next results, let us introduce the following symbols:
\begin{align*}
  &\overrightarrow{M}\setminus\{j\}\triangleq(M,\ldots,j+1,j-1,\ldots,2,1), \\
  &\overleftarrow{M}\setminus\{j\}\triangleq(1,2,\ldots,j-1,j+1,\ldots,M),\\
  &q_{\overrightarrow{M},x}\triangleq (q_{M,x},\ldots,q_{1,x}),\quad q_{\overleftarrow{M},x}\triangleq(q_{1,x},\ldots,q_{M,x}).
\end{align*}

\begin{corollary}\label{corollary1}
  Given a group of independent mKP eigenfunctions $f_1,f_2,\ldots,f_M$ and a group of independent mKP adjoint eigenfunctions $g_1,g_1,\ldots,g_N$, corresponding to the mKP tau pair $(\tau_0,\tau_1)$, if define $\tau_0^{[M+N]}$ and $\tau_1^{[M+N]}$ as follows,
  \begin{align*}
   \bullet ~ when ~ M \geq N, \nonumber \\
   \tau_0^{[M+N]} =& \frac{IW_{N,M}(g_{\overrightarrow{N},x};\ f_{\overleftarrow{M}}) \tau_1^{M-N}}{\tau_0^{M-N-1}},\quad \tau_1^{[M+N]} = \frac{IW_{N,M+1}(g_{\overrightarrow{N},x};\  f_{\overleftarrow{M}},1) \tau_1^{M-N+1}}{\tau_0^{M-N}},\\
  \bullet ~ when ~ M < N, \nonumber \\
   \tau_0^{[M+N]} =& \frac{IW_{M,N}(f_{\overrightarrow{M}};\ g_{\overleftarrow{N},x} ) \tau_1^{M-N}}{\tau_0^{M-N-1}},\quad
   \tau_1^{[M+N]} = \frac{(-1)^MIW_{M+1,N}(1,f_{\overrightarrow{M}};\ g_{\overleftarrow{N},x}) \tau_1^{M-N+1}}{\tau_0^{M-N}},
  \end{align*}
 then $\left(\tau_0^{[M+N]},\tau_1^{[M+N]}\right)$ is a new mKP tau pair.
\end{corollary}
\begin{remark}
Though there are similar results in \cite{cheng2018}, one can not know whether these tau functions satisfy the mKP bilinear equation \eqref{mkpbilnear2}. In fact, it is quite nontrivial to verify this. Please see Section \ref{section2} for more details.
\end{remark}
\begin{theorem}\label{th2}
 Given $f_i~ ( 1 \leq i \leq s )$ satisfying  $f_{i,t_n} = f_i^{(n)}$, $g_j ~( 1 \leq j \leq a )$ satisfying  $g_{j,t_n} = (-1)^{n-1}g_j^{(n)}$, $\varphi_l=e^{\xi(t,\lambda_l)}~(1\leq l \leq K)$ and $\psi_p=e^{-\xi(t,\mu_p)}~(1\leq p \leq Q)$, let us set
 \begin{align*}
   \mathfrak{q}_{\overrightarrow{M}}&\triangleq(\varphi_K,\ldots,\varphi_1, f_s^{(k M_s)},\ldots,f_s^{(k)},f_s,\ldots,f_1^{(k M_1)},\ldots,f_1^{(k)},f_1), \\
   \mathfrak{r}_{\overrightarrow{N}}&\triangleq(\psi_Q,\ldots, \psi_1,g_a^{(kN_a)},\ldots,g_a^{(k)},g_a,\ldots,
g_1^{(kN_1)},\ldots,g_1^{(k)},g_1),
 \end{align*}
where \( M = K + s + \sum_{j=1}^s M_j \), \(N = Q+a+\sum_{j=1}^a N_j\), $s+a=m$, and denote  $\tau_0^{[M+N]}$, $\tau_1^{[M+N]}$ in the following way.

$\bullet ~ When ~ M > N, $
\begin{align*}
   \tau_0^{[M+N]} &=IW_{N,M}(\mathfrak{r}_{\overrightarrow{N},x};\,\mathfrak{q}_{\overleftarrow{M}}),\quad \tau_1^{[M+N]} =IW_{N,M+1}(\mathfrak{r}_{\overrightarrow{N},x};\,\mathfrak{q}_{\overleftarrow{M}},1),\\
   \tau_{2,i}^{[M+N]} &=-IW_{N,M+1}(\mathfrak{r}_{\overrightarrow{N},x};\,\mathfrak{q}_{\overleftarrow{M}},f_i^{((M_i+1)k)}),\quad 1\leq i\leq s,\\
   \tau_{-1,i}^{[M+N]}&=(-1)^{M[i]+M}IW_{N,M}
    (\mathfrak{r}_{\overrightarrow{N},x};\,\mathfrak{q}_{\overleftarrow{M}\setminus\{M[i]\}},1),\quad 1\leq i\leq s,\\
   \tau_{2,j}^{[M+N]} &=(-1)^{N[i-s]+N}\frac{IW_{N,M}(1,\mathfrak{r}_{\overrightarrow{N}\setminus\{N[i-s]\}};\,\mathfrak{q}_{\overleftarrow{M},x})\cdot IW_{N,M+1}(\mathfrak{r}_{\overrightarrow{N},x};\,\mathfrak{q}_{\overleftarrow{M}},1)}
   {IW_{N,M}(\mathfrak{r}_{\overrightarrow{N}};\,\mathfrak{q}_{\overleftarrow{M},x})},\quad s+1\leq j\leq m,\nonumber\\
    \tau_{-1,j}^{[M+N]}&=(-1)^{M+k} \frac{IW_{N+1,M}(\mathfrak{r}_{\overrightarrow{N}},g_{i-s}^{((N_{i-s}+1)k)}; \mathfrak{q}_{\overleftarrow{M},x})\cdot IW_{N,M}(\mathfrak{r}_{\overrightarrow{N},x};\,\mathfrak{q}_{\overleftarrow{M}})}{IW_{N+1,M}(1,\mathfrak{r}_{\overrightarrow{N}};\, \mathfrak{q}_{\overleftarrow{M},x})},\quad s+1\leq j\leq m.
\end{align*}
Here we denote $M[i]=\sum_{j=1}^{i}M_j+i$ and $N[i]=\sum_{j=1}^{i}N_j+i$.

$\bullet ~ When ~ M =N,$
\begin{align*}
\tau_0^{[2M]} &=IW_{M,M}(\mathfrak{r}_{\overrightarrow{M},x};\,\mathfrak{q}_{\overleftarrow{M}}),\quad \tau_1^{[2M]} =IW_{M,M+1}(\mathfrak{r}_{\overrightarrow{M},x};\,\mathfrak{q}_{\overleftarrow{M}},1),\\
   \tau_{2,i}^{[2M]} &=-IW_{M,M+1}(\mathfrak{r}_{\overrightarrow{M},x};\,\mathfrak{q}_{\overleftarrow{M}},f_i^{((M_i+1)k)}),\quad 1\leq i\leq s,\\
   \tau_{-1,i}^{[2M]}&=(-1)^{M[i]+M}IW_{M,M}(1,\mathfrak{q}_{\overrightarrow{M}\setminus\{M[i]\}};\, \mathfrak{r}_{\overleftarrow{M},x}),\quad 1\leq i\leq s,\\
     \tau_{2,j}^{[2M]} &=(-1)^{M[i-s]+M}\frac{IW_{M,M}(1,\mathfrak{r}_{\overrightarrow{M}\setminus\{M[i-s]\}};\,\mathfrak{q}_{\overleftarrow{M},x})\cdot IW_{M,M+1}(\mathfrak{r}_{\overrightarrow{M},x};\,\mathfrak{q}_{\overleftarrow{M}},1)}
   {IW_{M,M}(\mathfrak{r}_{\overrightarrow{M}};\,\mathfrak{q}_{\overleftarrow{M},x})},\quad s+1\leq j\leq m,\\
    \tau_{-1,j}^{[2M]}&=(-1)^{k+1}\frac{IW_{M,M+1}(\mathfrak{q}_{\overrightarrow{M},x}; \mathfrak{r}_{\overleftarrow{M}},g_{i-s}^{((N_{i-s}+1)k)})\cdot IW_{M,M} (\mathfrak{r}_{\overrightarrow{M},x};\,\mathfrak{q}_{\overleftarrow{M}})}
    {IW_{M,M+1}({\mathfrak{q}_{\overrightarrow{M},x};\ \mathfrak{r}_{\overleftarrow{M}},1)}},\quad s+1\leq j\leq m.
   \end{align*}
   $\bullet ~ When ~ M < N,$
  \begin{align*}
   \tau_0^{[M+N]} &=  IW_{M,N}(\mathfrak{q}_{\overrightarrow{M}};\,\mathfrak{r}_{\overleftarrow{N},x}),\quad \tau_1^{[M+N]} = (-1)^{M}IW_{M+1,N}(1,\mathfrak{q}_{\overrightarrow{M}};\,\mathfrak{r}_{\overleftarrow{N},x}),\\
   \tau_{2,i}^{[M+N]} &= -IW_{M+1,N}(\mathfrak{q}_{\overrightarrow{M}},f_i^{((M_i+1)k)};\, \mathfrak{r}_{\overleftarrow{N},x}),\quad 1\leq i\leq s,\\
   \tau_{-1,i}^{[M+N]}&=(-1)^{M[i]+M}IW_{M,N}(1,\mathfrak{q}_{\overrightarrow{M}\setminus\{M[i]\}};\, \mathfrak{r}_{\overleftarrow{N},x}),\quad 1\leq i\leq s,\\
   \tau_{2,j}^{[M+N]} &=(-1)^{N[i-s]+M+N}\frac{IW_{M,N}(\mathfrak{q}_{\overrightarrow{M},x};\,\mathfrak{r}_{\overleftarrow{N}
   \setminus\{N[i-s]\}},1 )\cdot IW_{M+1,N}(1,\mathfrak{q}_{\overrightarrow{M}}; \mathfrak{r}_{\overleftarrow{N},x})}{IW_{M,N}(\mathfrak{q}_{\overrightarrow{M},x}; \mathfrak{r}_{\overleftarrow{N}})},\quad s+1\leq j\leq m,\\
\tau_{-1,j}^{[M+N]}&=(-1)^{k+1} \frac{IW_{M,N+1}(\mathfrak{q}_{\overrightarrow{M},x};\, \mathfrak{r}_{\overleftarrow{N}},g_{i-s}^{((N_{i-s}+1)k)} )\cdot IW_{M,N}(\mathfrak{q}_{\overrightarrow{M}};\, \mathfrak{r}_{\overleftarrow{N},x})}{IW_{M,N+1}(\mathfrak{q}_{\overrightarrow{M},x};\, \mathfrak{r}_{\overleftarrow{N}},1)},\quad s+1\leq j\leq m.
  \end{align*}
Then $\tau_0^{[M+N]}$,\ $\tau_1^{[M+N]}$,\ $\tau_{2,i}^{[M+N]}$ and $\tau_{-1,i}^{[M+N]}\ (1\leq i \leq m)$ satisfy the $(k,m)$-constrained mKP bilinear equations \eqref{kmcmkp1}-\eqref{kmcmkp3}.
\end{theorem}
\begin{remark}
The integral constants of $\Omega(f,g)$ in Theorem 2 are fixed by the way in \eqref{omeq} and \eqref{omer}, that is, the usual indefinite integral, if one of $f$ and $g$ is the sum of $e^{\xi(t,\lambda_i)}$. While if both $f$ and $g$ are polynomials of $t$, then $\Omega(f(t),g(t))=\sum_{n=0}^{+\infty}\int_0^1 t_i C_n(yt)dy$, where $C_n(t)=\Omega(f,g)_{t_n}=\mathrm{Res}_\pa (\pa^{-1}gA_nf\pa^{-1})$(please see \eqref{omegafg}).
\end{remark}
\subsection{The framework of this paper}
This paper is organized as follows.   In Section \ref{section2}, the relationship between the KP and mKP hierarchies is discussed, alongside the properties of the mKP Darboux transformations. Subsequently in Section \ref{section3}, we construct the generalized Wronskian solutions of the $(k, m)$-constrained mKP hierarchy. Next in Section \ref{section4}, we present several   examples of $(1, m)$-constrained and $(2, 2)$-constrained mKP hierarchies, deriving their evolution equations and explicit solutions. Finally in Section \ref{section5}, some conclusions and discussions are presented.

\section{The tau pair of the mKP hierarchy}\label{section2}
In this section, we first review the relations between the KP and mKP hierarchies in the aspects of tau functions, Lax operators, eigenfunctions, and so on. Then based on the KP Darboux transformation and its corresponding results, Theorem 1 and Corollary 1 are proved.

\subsection{The relation between the KP and mKP hierarchies}
Let us start from the KP hierarchy \cite{DJKM,Dickey2003} defined by the Lax equation $\mathcal{L}_{t_n}=[(\mathcal{L}^n)_{\geq 0},\mathcal{L}]$ with the Lax operator $\mathcal{L}=\pa+\sum_{j=1}^{\infty}u_{j+1}\pa^{-j}$, then we can also define the KP eigenfunction $\tilde{q}$ and the KP adjoint eigenfunction $\tilde{r}$ as follows:
\begin{align*}
\tilde{q}_{t_n}=(\mathcal{L}^n)_{\geq 0}(\tilde{q}),\quad \tilde{r}_{t_n}=-(\mathcal{L}^n)_{\geq 0}^*(\tilde{r}).
\end{align*}
\begin{proposition}\cite{cheng2018kp,Kac2018,Zabrodin2026}
Given the KP tau function $\tau_{\text {\tiny \rm{KP}}}$, let us assume that $\tilde{q}$ is the KP eigenfunction and $\tilde{r}$ is the KP
adjoint eigenfunction with respect to $\tau_{\text {\tiny \rm{KP}}}$, then
\begin{align*}
(\tau_0,\tau_1):=(\tau_{\text {\tiny \rm{KP}}}, \tilde{q}\tau_{\text {\tiny \rm{KP}}}) \quad or \quad (\tilde{r}\tau_{\text {\tiny \rm{KP}}}, \tau_{\text {\tiny \rm{KP}}}),
\end{align*}
satisfies the mKP bilinear equation \eqref{mkpbilnear2}. In particular, both $\tilde{q}\tau_{\text {\tiny \rm{KP}}}$ and $\tilde{r}\tau_{\text {\tiny \rm{KP}}}$ are new KP tau functions.

Conversely, if $(\tau_0,\tau_1)$ is the mKP tau pair, then $\tau_0$ and $\tau_1$ are the KP tau functions. And if set $\tilde{q}=\frac{\tau_1}{\tau_0},\ \tilde{r}=\frac{\tau_0}{\tau_1}$,
then $\tilde{q}$ is the KP eigenfunction corresponding to $\tau_0$, and $\tilde{r}$ is the KP adjoint eigenfunction corresponding to $\tau_1$.
\end{proposition}

\begin{remark}
  Here we say that the KP eigenfunction $\tilde{q}$ and the adjoint eigenfunction $\tilde{r}$ are corresponding to the KP tau function $\tau_{\text {\tiny \rm{KP}}}$, meaning that $\tilde{q}_{t_n}=(\mathcal{L}^n)_{\geq 0}(\tilde{q})$ and  $\tilde{r}_{t_n}=-(\mathcal{L}^n)_{\geq 0}^*(\tilde{r})$, where $\mathcal{L}=S\pa S^{-1}$ and $S=\sum_{i\geq 0}\frac{h_i(-\tilde{\pa})\tau_{\rm{KP}}(t)}{\tau_{\rm{KP}}(t)}\pa^{-i}$. Here $\tilde{\pa}=(\pa_x,\pa_{t_2}/2,\pa_{t_3}/3,\ldots)$ and $\exp(\xi(t,z))=\sum_{j=0}^{+\infty}h_j(t)z^j$.
\end{remark}

\begin{proposition}\cite{Shaw1997,Yang2022}
  Given the KP Lax operator $\mathcal{L}$, the KP eigenfunctions $\tilde{\phi}$ and $\tilde{q}$, and the KP adjoint eigenfunctions $\tilde{\psi}$ and $\tilde{r}$, we have the following results.

   If set $L=\tilde{\phi}^{-1}\mathcal{L}\tilde{\phi},\ q=\tilde{\phi}^{-1}\tilde{q},\ r=\Omega(\tilde{\phi}, \tilde{r})$, or
     $L=\pa^{-1}\tilde{\psi}\mathcal{L}\tilde{\psi}^{-1}\pa,\ q=\Omega(\tilde{q}, \tilde{\psi}),\ r=\tilde{\psi}^{-1}\tilde{r}$,
    then $L$ is the mKP Lax operator, $q$ is the mKP eigenfunction and $r$ is the mKP adjoint eigenfunction, satisfying $q_{t_n}=(L^n)_{\geq 1}(q)$ and $r_{t_n}=-(\pa(L^n)_{\geq 1}\pa^{-1})^*(r)$.
 \end{proposition}

\begin{proposition}\cite{Shaw1997,Yang2022}
  Given the mKP Lax operator $L=\pa+\sum_{i=0}^{\infty}v_{i}\pa^{-i}$, the mKP tau pair $(\tau_0,\tau_1)$, the mKP eigenfunction $q$ and the mKP adjoint eigenfunction $r$, we have the results below.

  $\bullet$ ~ $\textbf{Case 1:}$
   If set
   \begin{align*}
   \mathcal{L}=e^{\int v_0 dx}Le^{-\int v_0 dx},\quad \tilde{q}=e^{\int v_0 dx}q, \quad \tilde{r}= e^{-\int v_0 dx}r_x,
   \end{align*}
  then $\mathcal{L}$ is the KP Lax operator, $\tilde{q}$ is the KP eigenfunction and $\tilde{r}$ is the KP adjoint eigenfunction, corresponding to the KP tau function $\tau_0$. In particular, $e^{\int v_0 dx}$ is also the corresponding KP eigenfunction, since 1 is the mKP eigenfunction.

    $\bullet$ ~ $\textbf{Case 2:}$
  If set
   \begin{align*}
   \mathcal{L}=e^{\int v_0 dx}\pa L \pa^{-1} e^{-\int v_0 dx},\quad \tilde{q}=e^{\int v_0 dx}q_x, \quad \tilde{r}= e^{-\int v_0 dx}r,
   \end{align*}
  then $\mathcal{L}$ is the KP Lax operator, $\tilde{q}$ is the KP eigenfunction and $\tilde{r}$ is the KP adjoint eigenfunction, corresponding to the KP tau function $\tau_1$. In particular, $e^{-\int v_0 dx}$ is also the corresponding KP adjoint eigenfunction, since 1 is the mKP adjoint eigenfunction.
\end{proposition}

\subsection{KP Darboux transformations}
In what follows, we present some fundamental facts concerning the KP Darboux transformations (please refer to \cite{Yang2022,He2002,Zabrodin225}).
\begin{proposition}\cite{Yang2022,He2002,Zabrodin225}
Given the KP tau function $\tau_{\text {\tiny \rm{KP}}}$, the KP Lax operator $\mathcal{L}$, the KP eigenfunction $\tilde{\phi}$ and $\tilde{q}$, and the KP adjoint eigenfunction $\tilde{\psi}$ and $\tilde{r}$,

$\bullet$ ~ $\textbf{Case 1:}$
 when $\tau_{\text {\tiny \rm{KP}}} \to \tau_{\text {\tiny \rm{KP}}}^{[1]}=\tilde{\phi}\tau_{\text {\tiny \rm{KP}}}$,
 \begin{align*}
   L \to L^{[1]}=&T_d(\tilde{\phi})LT_d(\tilde{\phi})^{-1},\quad T_d(\tilde{\phi})=\tilde{\phi}\pa \tilde{\phi}^{-1},\\
   \tilde{q} \to \tilde{q}^{[1]}=&T_d(\tilde{\phi})(\tilde{q}), \quad \tilde{r} \to \tilde{r}^{[1]}=T_d(\tilde{\phi})^{-1*}(\tilde{r}).
 \end{align*}

$\bullet$ ~ $\textbf{Case 2:}$
 when $\tau_{\text {\tiny \rm{KP}}} \to \tau_{\text {\tiny \rm{KP}}}^{[1]}=\tilde{\psi}\tau_{\text {\tiny \rm{KP}}}$,
 \begin{align*}
    L \to L^{[1]}=&T_i(\tilde{\psi})LT_i(\tilde{\psi})^{-1}, \quad T_i(\tilde{\psi})=\tilde{\psi}^{-1}\pa^{-1}\tilde{\psi},\\
   \tilde{q} \to \tilde{q}^{[1]}=&T_i(\tilde{\psi})(\tilde{q}), \quad \tilde{r} \to \tilde{r}^{[1]}=T_i(\tilde{\psi})^{-1*}(\tilde{r}).
 \end{align*}
Here $L^{[1]}$, $q^{[1]}$ and $r^{[1]}$ are the new KP Lax operator, the new KP eigenfunction and the new KP adjoint eigenfunction respectively, corresponding to the new KP tau function $\tau_{\text {\tiny \rm{KP}}}^{[1]}$.
\end{proposition}

\begin{proposition}\cite{Yang2022,He2002}
    Given a group of independent KP eigenfunctions $\phi,\phi_1,\ldots,\phi_M$ and a group of independent KP adjoint eigenfunctions $\psi,\psi_1,\ldots,\psi_N$, corresponding to the KP tau function $\tau_{\text {\tiny \rm{KP}}}$,
    if denote $\phi^{[M+N]}$, $\psi^{[M+N]}$ and $\tau^{[M+N]}$ as follows:
    \begin{itemize}
\item when $M > N$,
     \begin{align*}
        \phi^{[M+N]} =& \frac{IW_{N,M+1}(\psi_{\overrightarrow{N}}; \phi_{\overleftarrow{M}}, \phi)}{IW_{N,M}(\psi_{\overrightarrow{N}}; \phi_{\overleftarrow{M}})},\quad \psi^{[M+N]} = \frac{(-1)^M IW_{N+1,M}(\psi, \psi_{\overrightarrow{N}}; \phi_{\overleftarrow{M}})}{IW_{N,M}(\psi_{\overrightarrow{N}}; \phi_{\overleftarrow{M}})},\\ \tau^{[M+N]}=&(-1)^{NM}IW_{N,M}(\psi_{\overrightarrow{N}};\phi_{\overleftarrow{M}})\tau;
     \end{align*}

\item when $M = N$,
     \begin{align*}
     \phi^{[2M]} &= \frac{IW_{M,M+1}(\psi_{\overrightarrow{M}}; \phi_{\overleftarrow{M}},\phi)}{IW_{M,M}(\psi_{\overrightarrow{M}}; \phi_{\overleftarrow{M}})}, \quad
     \psi^{[2M]} = \frac{(-1)^M IW_{M,M+1}(\phi_{\overleftarrow{M}}; \psi,\psi_{\overrightarrow{M}})}{IW_{M,M}(\psi_{\overrightarrow{M}};\phi_{\overleftarrow{M}} )}, \\
     \tau^{[2M]} &= (-1)^M IW_{M,M}(\psi_{\overrightarrow{M}}; \phi_{\overleftarrow{M}})\tau;
     \end{align*}

\item when $M < N$,
     \begin{align*}
        \phi^{[M+N]} =& \frac{(-1)^M IW_{M+1,N}(\phi,\phi_{\overrightarrow{M}}; \psi_{\overleftarrow{N}})}{IW_{M,N}(\phi_{\overrightarrow{M}}; \psi_{\overleftarrow{N}})},\quad
       \psi^{[M+N]} = \frac{(-1)^{M+N} IW_{M,N+1}(\phi_{\overrightarrow{M}}; \psi_{\overleftarrow{N}}, \psi)}{IW_{M,N}(\phi_{\overrightarrow{M}}; \psi_{\overleftarrow{N}})} \\
       \tau^{[M+N]}=&(-1)^{MN+\frac{M(M-1)}{2}+\frac{N(N-1)}{2}} IW_{M,N}(\phi_{\overleftarrow{M}}; \psi_{\overrightarrow{N}})\tau;
     \end{align*}
     \end{itemize}
then $ \tau^{[M+N]}$ is the new KP tau function, $\phi^{[M+N]}$ is the corresponding KP eigenfunction and $\psi^{[M+N]}$ is the corresponding KP adjoint eigenfunction.
\end{proposition}

\subsection{Proofs of Theorem 1 and Corollary 1}
With the preparation above, we now begin to \textbf{prove Theorem 1}. We only prove the case of $(\tau_0^{[1]},\tau_1^{[1]})=(q \tau_1,-q_x\frac{\tau_1^2}{\tau_0})$, since the other one can be proved similarly.
Firstly for the mKP Lax operator $L=\pa+\sum_{i=0}^{\infty}v_{i}\pa^{-i}$, we can get  $e^{\int v_0 dx}=\frac{\tau_1}{\tau_0}\triangleq \tilde{\phi}$ by \eqref{adwavefunction}. By Proposition 3, if set $\tilde{q}=\tilde{\phi}q$, then $\tilde{\phi}$ and $\tilde{q}$ are the KP eigenfunctions, corresponding to the KP Lax operator $L=\tilde{\phi}\mathcal{L}\tilde{\phi}^{-1}$ and the KP tau function $\tau_0$ $(\tau_0 \triangleq \tau_{\text {\tiny \rm{KP}}})$. Further we can find
\begin{align*}
\tau_0^{[1]}=\tilde{q}\tau_{\text {\tiny \rm{KP}}},\quad \tau_1^{[1]}
= \left( \tilde{\phi}_x \tilde{q} - \tilde{\phi} \tilde{q}_x \right) \tau_{\text {\tiny \rm{KP}}} = \tilde{\phi}^{[1]} \tilde{q} \tau_{\text {\tiny \rm{KP}}},
\end{align*}
where $\tilde{\phi}^{[1]} = T_d(\tilde{q})(\tilde{\phi})$. By Proposition 4, we know $\tilde{\phi}^{[1]}$ is another KP eigenfunction, corresponding to the KP Lax operator $L^{[1]}=T_d(\tilde{\phi})LT_d^{-1}(\tilde{\phi})$.\ Therefore by Proposition 1,
$$\left( \tau_{\text {\tiny \rm{KP}}}^{[1]}, \tilde{\phi}^{[1]} \tau_{\text {\tiny \rm{KP}}}^{[1]} \right)=\left(\tilde{q}\tau_{\text {\tiny \rm{KP}}},\tilde{\phi}^{[1]} \tilde{q} \tau_{\text {\tiny \rm{KP}}}\right)$$
is the mKP tau pair. Thus, we complete $\textbf{the proof of Theorem 1}$.

Next, let us start by \textbf{proving Corollary 1}. We only prove the case of $M\geq N$, since the other one can be proved similarly. Firstly, note that $\tau_0$ is the KP tau function by Proposition 1, and let us denote it by $\tau_{\text {\tiny \rm{KP}}}$. If set $\tau_1=\tilde{\phi}\tau_{\text {\tiny \rm{KP}}}$, $\tilde{f_i}=\tilde{\phi}f_i$ and $\tilde{g_i}=\tilde{\phi}^{-1}g_{i,x}$, then by Proposition 4, we can find that $\tilde{f_i}$ and $\tilde{g_i}$ are the KP eigenfunctions and the KP adjoint eigenfunctions respectively, corresponding to $\tau_{\text {\tiny \rm{KP}}}$, and $\tilde{\phi}$ is also the KP eigenfunction corresponding to $\tau_{\text {\tiny \rm{KP}}}$. Further we have:
\begin{align}
 \tau_0^{[M+N]} =& \frac{IW_{N,M}(g_{\overrightarrow{N},x}; f_{\overleftarrow{M}}) \tau_1^{M-N}}{\tau_0^{M-N-1}}
 =IW_{N,M}(\tilde{\phi}\widetilde{g}_N,  \ldots, \tilde{\phi}\widetilde{g}_1; \tilde{\phi}^{-1}\widetilde{f}_1, \ldots, \tilde{\phi}^{-1}\widetilde{f}_M)\tilde{\phi}^{M-N}\tau_{\text {\tiny \rm{KP}}} \nonumber \\
=&\begin{vmatrix}
\Omega(\widetilde{f}_1, \widetilde{g}_N) & \Omega(\widetilde{f}_2, \widetilde{g}_N) & \cdots & \Omega(\widetilde{f}_M, \widetilde{g}_N) \\
\vdots & \vdots & \ddots & \vdots \\
\Omega(\widetilde{f}_1, \widetilde{g}_1) & \Omega(\widetilde{f}_2, \widetilde{g}_1) & \cdots & \Omega(\widetilde{f}_M, \widetilde{g}_1) \\
\tilde{\phi}^{-1}\widetilde{f}_1 & \tilde{\phi}^{-1}\widetilde{f}_2 & \cdots & \tilde{\phi}^{-1}\widetilde{f}_M \\
\left(\tilde{\phi}^{-1}\widetilde{f}_1\right)_x & \left(\tilde{\phi}^{-1}\widetilde{f}_2\right)_x & \cdots & \left(\tilde{\phi}^{-1}\widetilde{f}_M\right)_x \\
\vdots & \vdots & \ddots & \vdots \\
\left(\tilde{\phi}^{-1}\widetilde{f}_1\right)^{(M-N-1)} & \left(\tilde{\phi}^{-1}\widetilde{f}_2\right)^{(M-N-1)} & \cdots & \left(\tilde{\phi}^{-1}\widetilde{f}_M\right)^{(M-N-1)}
\end{vmatrix}\cdot \tilde{\phi}^{M-N}\tau_{\text {\tiny \rm{KP}}}.\label{coropf1}
\end{align}
Note that $\left(\tilde{\phi}^{-1}\widetilde{f}_i\right)^{(j)}=\sum_{k=0}^{j}C_j^k(\tilde{\phi}^{-1})^{(j-k)}\tilde{f}_i^{(k)}$ with $C_j^k=\frac{j(j-1)\cdots(j-k+1)}{k!}$. So if we apply elementary row operations to the determinant in \eqref{coropf1}, then we can find
\begin{align*}
\tau_0^{[M+N]}=\begin{vmatrix}
\Omega(\widetilde{f}_1, \widetilde{g}_N) & \Omega(\widetilde{f}_2, \widetilde{g}_N) & \cdots & \Omega(\widetilde{f}_M, \widetilde{g}_N) \\
\vdots & \vdots & \ddots & \vdots \\
\Omega(\widetilde{f}_1, \widetilde{g}_1) & \Omega(\widetilde{f}_2, \widetilde{g}_1) & \cdots & \Omega(\widetilde{f}_M, \widetilde{g}_1) \\
\tilde{\phi}^{-1}\widetilde{f}_1 & \tilde{\phi}^{-1}\widetilde{f}_2 & \cdots & \tilde{\phi}^{-1}\widetilde{f}_M \\
\tilde{\phi}^{-1}\widetilde{f}_{1,x} & \tilde{\phi}^{-1}\widetilde{f}_{2,x} & \cdots & \tilde{\phi}^{-1}\widetilde{f}_{M,x} \\
\vdots & \vdots & \ddots & \vdots \\
\tilde{\phi}^{-1}\widetilde{f}_1^{(M-N-1)} & \tilde{\phi}^{-1}\widetilde{f}_2^{(M-N-1)} & \cdots & \tilde{\phi}^{-1}\widetilde{f}_M^{(M-N-1)}
\end{vmatrix} \cdot \tilde{\phi}^{M-N} \tau_{\text {\tiny \rm{KP}}}=IW_{N,M}(\widetilde{g}_{\overrightarrow{N}}; \widetilde{f}_{\overleftarrow{M}})\tau_{\text {\tiny \rm{KP}}}.
\end{align*}
A similar discussion can lead to
\begin{align*}
 \tau_1^{[M+N]}=&IW_{N,M+1}(\widetilde{g}_{\overrightarrow{N}}; \widetilde{f}_{\overleftarrow{M}},\tilde{\phi})\tau_{\text {\tiny \rm{KP}}}.
\end{align*}
Furthermore by Proposition 5, we can know
 \begin{align*}
 \tilde{\phi}^{[M+N]} =& \frac{IW_{N,M+1}(\widetilde{g}_{\overrightarrow{N}}; \widetilde{f}_{\overleftarrow{M}},\tilde{\phi})}{IW_{N,M}(\widetilde{g}_{\overrightarrow{N}}; \widetilde{f}_{\overleftarrow{M}})},
\end{align*}
is the KP eigenfunction corresponding to $\tau_0^{[M+N]}$.
Therefore
\begin{align*}
 \tau_1^{[M+N]}= \tilde{\phi}^{[M+N]}\tau_0^{[M+N]},
\end{align*}
and by Proposition 1, we complete $\textbf{the proof of Corollary 1}$.

\section{The constrained modified KP hierarchy}\label{section3}
In this section, we first review the mKP Darboux transformations and apply them to the constrained mKP hierarchy. On this basis, we derive the generalized Wronskian solutions for the constrained mKP hierarchy, satisfying the corresponding bilinear equations, that is Theorem 2.
\subsection{The mKP Darboux transformations}
Let us review the mKP Darboux transformation \cite{Oevel1998JP,cheng2018,Yang2022}. One pseudo-differential operator $T$ is called the mKP Darboux operator if for the mKP Lax operator $L$,
\begin{align*}
L^{[1]}\triangleq TLT^{-1}
\end{align*}
is still the mKP Lax operator. It can be proved that there are two kinds of basic mKP Darboux operators for the mKP eigenfunction $f$ and the mKP adjoint eigenfunction $g$:
\begin{align}\label{TdTi}
T_D(f) \triangleq (f^{-1})_x^{-1}\pa f^{-1},\quad T_I(g) \triangleq g^{-1}\pa^{-1}g_x,
\end{align}
which can commute with each other, meaning that
\begin{align*}
&T_D(f_2^{[1]})T_D(f_1)=T_D(f_1^{[1]})T_D(f_2),\ \textup{with} \ f_i^{[1]}=T_D(f_{3-i})(f_i),\ i=1 \ \textup{or}\ 2,  \\
&T_I(g_2^{[1]})T_I(g_1)=T_I(g_1^{[1]})T_I(g_2),\ \textup{with} \ g_i^{[1]}=\left(\pa^{-1}T_I(g_{3-i})^{-1*}\pa\right)(g_i), \ i=1 \ \textup{or}\ 2,\\
&T_D(f^{[1]})T_I(g)=T_I(g^{[1]})T_D(f),\ \textup{with} \ f^{[1]}=T_I(g)(f), \ g^{[1]}=\left(\pa^{-1}T_D(f)^{-1*}\pa\right)(g).
\end{align*}

Due to the above commutativity of $T_D$ and $T_I$, we can only consider the following mKP Darboux chain:
\begin{eqnarray*}
&&L^{[0]}\xrightarrow{T_D(f_1^{[0]})}L^{[1]}\xrightarrow{T_D(f_2^{[1]})}L^{[2]}
  \rightarrow\cdots\rightarrow L^{[M-1]}\xrightarrow{T_D(f_M^{[M-1]})}L^{[M]}\\
&&\xrightarrow{T_I(g_1^{[M]})}L^{[M+1]}\xrightarrow{T_I(g_2^{[M+1]})}\cdots\rightarrow L^{[M+N-1]}\xrightarrow{T_I(g_N^{[M+N-1]})}L^{[M+N]},
\end{eqnarray*}
where $f_i~(1\leq i\leq M)$ and  $g_j~(1\leq j\leq N)$  are the mKP eigenfunctions and the mKP adjoint eigenfunctions respectively. Here $f_i^{[i-1]} $, $L^{[i]}~(1 \leq i\leq M)$, $g_j^{[M+j-1]}$ and $L^{[M+j]}~(1 \leq j \leq N)$ are given by
$f_i^{[i-1]}\triangleq T^{[\overrightarrow{i-1},\overrightarrow{0}]}(f_i)$ ,
$L^{[i]}\triangleq T^{[\overrightarrow{i},\overrightarrow{0}]}\cdot L\cdot T^{[\overrightarrow{i},\overrightarrow{0}]-1}$,
$g_j^{[M+j-1]}\triangleq \left(\pa^{-1}(T^{[\overrightarrow{M},\overrightarrow{j-1}]-1})^*\pa\right)(g_j)$ and
$L^{[M+j]}\triangleq T^{[\overrightarrow{M},\overrightarrow{j}]}\cdot L\cdot T^{[\overrightarrow{M},\overrightarrow{j}]-1}$,
where
\begin{align}
T^{[\overrightarrow{M},\overrightarrow{N}]}\triangleq T^{[\overrightarrow{M},\overrightarrow{N}]}(f_{\overleftarrow{M}};g_{\overleftarrow{N}})=T_I(g_N^{[M+N-1]})\ldots
T_I(g_1^{[M]})T_D(f_M^{[M-1]})\ldots T_D(f_1).\label{Tnkexpr}
\end{align}
 Here $T^{[\overrightarrow{M},\overrightarrow{N}]}$ and $\pa^{-1}\left({T}^{[\overrightarrow{M},\overrightarrow{N}]-1} \right)^*\pa$ have determinant representations \cite{cheng2018}, which can be found in the appendix of this paper.
\subsection{The Darboux transformations for the constrained mKP hierarchy}
If we apply the mKP Darboux operators $T_D(f)$ and $T_I(g)$ to the $(k,m)$-constrained mKP hierarchy:
\begin{align}
&(L^k)_{\leq 0}=\sum_{i=1}^{m}q_i\pa^{-1}r_i\pa,\quad L_{t_n}=[(L^n)_{\geq 1},L],\label{Lk<0} \\
&q_{j,t_n}=(L^n)_{\geq1}(q_j),\quad r_{j,t_n}=-\big(\pa^{-1}(L^n)^*_{\geq1}\pa\big)(r_j), \quad 1\leq j\leq m,\label{qjrj}
\end{align}
we can get\cite{Chen20193,Yangyi2020}
\begin{align}\label{Lk[1]}
(L^{[1]})^k_{\leq 0}=q_0^{[1]}\pa^{-1}r_0^{[1]}\pa+\sum_{i=1}^{m}q_i^{[1]}\pa^{-1}r_i^{[1]}\pa.
\end{align}
Here $L^{[1]}=TLT^{-1}$, where $T=T_D(f)$ or $T_I(g)$.

 $\bullet$ ~ $\textbf{Case 1:} ~ T=T_D(f)$,
\begin{align}
  &q_0^{[1]}=-(T_D(f)L^k)(f),\ \ \  r_0^{[1]}=f^{-1},\label{TD(f)q0} \\
  &q_i^{[1]}=T_D(f)(q_i),\ \ \  r_i^{[1]}=(\pa T_D(f)^{-1} \pa^{-1})^*(r_i),\label{TD(f)qi}
\end{align}

 $\bullet$ ~ $\textbf{Case 2:} ~ T=T_I(g)$,
\begin{align}\label{TI(g)}
  &q_0^{[1]}=g^{-1},\ \ \ r_0^{[1]}=-(\pa L^kT_I(g)^{-1}\pa^{-1})^*(g), \\
  &q_i^{[1]}=T_I(g)(q_i),\ \ \ r_i^{[1]}=(\pa T_I(g)^{-1} \pa^{-1})^*(r_i).
\end{align}
To compute $L^{[M+N]}=T^{[\overrightarrow{M},\overrightarrow{N}]}\cdot L\cdot T^{[\overrightarrow{M},\overrightarrow{N}]-1}$ for the Lax operator $L$ of the $(k,m)$-constrained mKP hierarchy, we need the following lemma.
\begin{lemma}
The following relations hold:
\begin{align}
&\left( \pa^{-1} (T_D(f_2^{[1]})^{-1})^* \pa \right)(f_1^{-1}) =\Bigl( T_D(f_2)(f_1) \Bigr)^{-1},\label{fourequations1} \\
&\left( \pa^{-1} (T_I(g_2^{[1]})^{-1})^* \pa \right)(f_1^{-1}) =\Bigl( T_I(g_2)(f_1) \Bigr)^{-1},\label{fourequations2} \\
&T_D(f_2^{[1]})(g_1^{-1})= \Bigl(\left(\pa^{-1} (T_D(f_2)^{-1})^* \pa \right)(g_1)\Bigr)^{-1},\label{fourequations3} \\
&T_I(g_2^{[1]})(g_1^{-1})= \Bigl(\left(\pa^{-1} (T_I(g_2)^{-1})^* \pa \right)(g_1)\Bigr)^{-1}.\label{fourequations4}
\end{align}
Here $f_2^{[1]}=T_D(f_1)(f_2)$ in \eqref{fourequations1}, $g_2^{[1]}=\Bigl(\pa^{-1} (T_D(f_1)^{-1})^* \pa \Bigr)(g_2)$ in \eqref{fourequations2}, $f_2^{[1]}=T_I(g_1)(f_2)$ in \eqref{fourequations3}, and $g_2^{[1]}=\Bigl(\pa^{-1} (T_I(g_1)^{-1})^* \pa \Bigr)(g_2)$ in \eqref{fourequations4}.
\end{lemma}
\begin{proof}
We only prove \eqref{fourequations1}, since the others can be proved similarly. To prove \eqref{fourequations1}, we only need to check
\begin{align}\label{fourequations11}
\left(T_D(f_2^{[1]})^{-1}\right)^*\Bigl((f_1^{-1})_x\Bigr)=\left(\bigl( T_D(f_2)(f_1) \bigr)^{-1}\right)_x,
\end{align}
where $f_2^{[1]}=(f_1^{-1}f_2)_x/(f_1^{-1})_x$. By the direct computation, we can find that both sides of \eqref{fourequations11} are equal to $-f_1^{-1}f_2(f_2^{[1]-1})_x$.
\end{proof}

\begin{proposition}
  For the Lax operator $L$ of the $(k,m)$-constrained mKP hierarchy \eqref{Lk<0}\eqref{qjrj}:
  \begin{equation}\label{LMNk}
  \begin{aligned}
    (L^{[M+N]})^k_{<1}
    =&-\sum_{i=1}^{N}\left(\bigl(\pa^{-1}(T^{[\overrightarrow{M},\overrightarrow{N}\setminus \{i\}]-1})^* \pa\bigr)(g_i)\right)^{-1} \cdot \pa^{-1} \cdot \left(\pa^{-1}\bigl(T^{[\overrightarrow{M},\overrightarrow{N}]-1}\bigr)^*(L^k)^*\pa\right)(g_i)\cdot\pa\\
    &-\sum_{i=1}^{M}\left(T^{[\overrightarrow{M},\overrightarrow{N}]}L^k \right)(f_i)\cdot \pa^{-1}\cdot
    \left((T^{[\overrightarrow{M}\setminus \{i\},\overrightarrow{N}]}) (f_i)\right)^{-1}\cdot \pa \\
    &+\sum_{i=1}^{m}T^{[\overrightarrow{M},\overrightarrow{N}]}(q_i)\cdot\pa^{-1}\cdot
     \left(\pa^{-1}\bigl(T^{[\overrightarrow{M},
     \overrightarrow{N}]-1}\bigr)^*\pa\right)(r_i)\cdot\pa.
 \end{aligned}
  \end{equation}
\end{proposition}

\begin{proof}
  Obviously when $M=N=0$, the conclusion holds. Suppose that the conclusion holds for $M+N$.
  We next prove that \eqref{LMNk} also holds for $M+N+1$. Now there are two cases for obtaining $L^{[M+N+1]}$ from $L^{[M+N]}$. The first one is given by $T^{[\overrightarrow{M+1},\overrightarrow{N}]}$, another one is the usage of $T^{[\overrightarrow{M},\overrightarrow{N+1}]}$. Here we only prove the case of $T^{[\overrightarrow{M+1},\overrightarrow{N}]}$, since the case of $T^{[\overrightarrow{M},\overrightarrow{N+1}]}$ can be proved similarly.

  Notice that in the case of $T^{[\overrightarrow{M+1},\overrightarrow{N}]}$, we have $L^{[M+N+1]}=T_D(f_{M+1}^{[M+N]})L^{[M+N]}T_D(f_{M+1}^{[M+N]})^{-1}$. Thus by \eqref{Lk[1]}-\eqref{TD(f)qi},
  \begin{small}
  \begin{align}
  &(L^{[M+N+1]})^k_{<1}=-\left(T_D(f_{M+1}^{[M+N]})L^{[M+N]k}\right)(f_{M+1}^{[M+N]})\cdot \pa^{-1} \cdot (f_{M+1}^{[M+N]-1})\pa \nonumber \\
  &-\sum_{i=1}^{N}T_D(f_{M+1}^{[M+N]})\left(\bigl(\pa^{-1}(T^{[\overrightarrow{M},\overrightarrow{N}\setminus \{i\}]-1})^* \pa\bigr)(g_i^{-1})\right) \cdot \pa^{-1}\nonumber\\
  &\cdot \left(\pa^{-1}\bigl(T_D(f_{M+1}^{[M+N]})^{-1}\bigr)^*\pa\right)
  \left(\pa^{-1}(T^{[\overrightarrow{M},\overrightarrow{N}]-1})^*(L^k)^*\pa\right)(g_i)\cdot\pa \nonumber\\
  &-\sum_{i=1}^{M}T_D(f_{M+1}^{[M+N]})\left(T^{[\overrightarrow{M},\overrightarrow{N}]}L^k \right)(f_i)\cdot \pa^{-1}
  \cdot \left(\pa^{-1}\bigl(T_D(f_{M+1}^{[M+N]})^{-1}\bigr)^*\pa\right)
    \left((T^{[\overrightarrow{M}\setminus \{i\},\overrightarrow{N}]}) (f_i^{-1})\right)\cdot \pa \nonumber\\
  &+\sum_{i=1}^{m}T_D(f_{M+1}^{[M+N]})T^{[\overrightarrow{M},\overrightarrow{N}]}(q_i)\cdot\pa^{-1}
  \cdot\left(\pa^{-1}\bigl(T_D(f_{M+1}^{[M+N]})^{-1}\bigr)^*\pa\right)
  \left(\pa^{-1}(T^{[\overrightarrow{M},\overrightarrow{N}]-1})^*\pa\right)
  (r_i)\cdot\pa.\label{LMNk11}
  \end{align}
\end{small}

  Note that $T^{[\overrightarrow{M+1},\overrightarrow{N}]}=T_D(f_{M+1}^{[M+N]})T^{[\overrightarrow{M},\overrightarrow{N}]}$
  and $T^{[\overrightarrow{M},\overrightarrow{N}]}=T^{[\overrightarrow{M+1}\setminus \{M+1\},\overrightarrow{N}]}$, then the first item on the RHS of \eqref{LMNk11} becomes
  \begin{align*}
   -\left(T^{[\overrightarrow{M+1},\overrightarrow{N}]}L^k\right)(f_{M+1})\cdot \pa^{-1} \cdot
  \left(T^{[\overrightarrow{M+1}\setminus \{M+1\},\overrightarrow{N}]}(f_{M+1})\right)^{-1}\cdot\pa.
  \end{align*}
   Next by \eqref{fourequations3}, the second item on the RHS of \eqref{LMNk11} becomes
  $$-\sum_{i=1}^{N}\left(\bigl(\pa^{-1}(T^{[\overrightarrow{M+1},\overrightarrow{N}\setminus \{i\}]-1})^* \pa\bigr)(g_i)\right)^{-1}\cdot \pa^{-1} \cdot \left(\pa^{-1}(T^{[\overrightarrow{M+1},\overrightarrow{N}]-1})^*(L^k)^*\pa\right)(g_i)\cdot \pa.$$
  Similarly by $T^{[\overrightarrow{M+1},\overrightarrow{N}]}=T_D(f_{M+1}^{[M+N]})T^{[\overrightarrow{M},\overrightarrow{N}]}$, the third item and the fourth item on the RHS of \eqref{LMNk11} become
    \begin{align*}
    -\sum_{i=1}^{M}\left(T^{[\overrightarrow{M+1},\overrightarrow{N}]}L^k \right)(f_i)\cdot \pa^{-1}\cdot  \left((T^{[\overrightarrow{M+1}\setminus \{i\},\overrightarrow{N}]}) (f_i)\right)^{-1}\cdot \pa,
  \end{align*}
  and
  $$\sum_{i=1}^{m}T^{[\overrightarrow{M+1},\overrightarrow{N}]}(q_i)\cdot\pa^{-1}\cdot
     \left(\pa^{-1}(T^{[\overrightarrow{M+1},\overrightarrow{N}]-1})^*\pa\right)(r_i)\cdot\pa.$$
  Based on the above relations, this proposition can be proved.
\end{proof}

\begin{corollary}\label{Corollary2}
 Under the same conditions of Theorem 2, we have
\begin{small}
  \begin{align}
 &(T^{[\overrightarrow{M},\overrightarrow{N}]}\cdot\pa^k\cdot T^{[\overrightarrow{M},\overrightarrow{N}]-1})_{<1}= -\sum_{i=1}^{s}  T^{[\overrightarrow{M},\overrightarrow{N}]} (f_i^{((M_i+1)k)})\cdot \pa^{-1} \cdot \biggl( T^{[\overrightarrow{M}\setminus \{M[i]\},\overrightarrow{N}]}(f_i^{(M_i k)})\biggr)^{-1}\cdot \pa \nonumber \\
&- (-1)^k\sum_{i=1}^{a} \Biggl( \left(  \pa^{-1}  T^{[\overrightarrow{M},\overrightarrow{N}\setminus \{N[i]\}]-1*} \pa  \right) (g_i^{(N_ik)}) \Biggr)^{-1} \cdot \pa^{-1} \cdot \left(\pa^{-1}T^{[\overrightarrow{M},\overrightarrow{N}]-1*}\pa\right)(g_i^{((N_i+1) k)})\cdot\pa.\label{LMNKnew1}
\end{align}
  \end{small}
Here we denote $M[i]=\sum_{j=1}^{i}M_j+i$ and $N[i]=\sum_{j=1}^{i}N_j+i$.
\end{corollary}
\begin{proof}
First by Proposition 6,
\begin{equation*}
  {
  \begin{aligned}
&(T^{[\overrightarrow{M},\overrightarrow{N}]}\cdot\pa^k\cdot T^{[\overrightarrow{M},\overrightarrow{N}]-1})_{<1}=\\
&-\sum_{i=1}^{s}  \sum_{l=0}^{M_i}  T^{[\overrightarrow{M},\overrightarrow{N}]} (f_i^{((l+1)k)})\cdot \pa^{-1} \cdot \biggl( T^{[\overrightarrow{M}\setminus \{M[i]-M_i+l\},\overrightarrow{N}]}(f_i^{(l k)})\biggr)^{-1}\cdot \pa  \\
&- \sum_{l=1}^{K}  T^{[\overrightarrow{M},\overrightarrow{N}]} \pa^k  (\varphi_l)\cdot \pa^{-1} \cdot \left( T^{[\overrightarrow{M}\setminus \{M[s]+l\},\overrightarrow{N}]} (\varphi_l)\right)^{-1}\cdot \pa \\
&- (-1)^k\sum_{i=1}^{a}  \sum_{p=0}^{N_i} \Biggl( \left(  \pa^{-1} T^{[\overrightarrow{M},\overrightarrow{N}\setminus \{N[i]-N_i+p\}]-1*}  \pa  \right) (g_i^{(p k)}) \Biggr)^{-1} \cdot \pa^{-1} \cdot
\left(\pa^{-1}T^{[\overrightarrow{M},\overrightarrow{N}]-1*}\pa\right)(g_i^{((p+1) k)})\cdot\pa\\
&-(-1)^k\sum_{p=1}^{Q}\Biggl(\left(  \pa^{-1}  T^{[\overrightarrow{M},\overrightarrow{N}\setminus \{N[a]+p\}]-1*}  \pa  \right)(\psi_p)\Biggr)^{-1} \cdot \pa^{-1} \cdot \left(\pa^{-1}T^{[\overrightarrow{M},\overrightarrow{N}]-1*}\pa^{k+1}\right)(\psi_p)\cdot \pa.
\end{aligned}
  }\label{LMNKnew}
  \end{equation*}
Note that
 $$T^{[\overrightarrow{M},\overrightarrow{N}]} (f_i^{((\alpha+1)k)})=0 \ (\alpha\neq M_i),\quad \left(\pa^{-1}T^{[\overrightarrow{M},\overrightarrow{N}]-1*}\pa\right)(g_i^{((\beta+1) k)})=0 \ (\beta\neq N_i),$$
and
$$ T^{[\overrightarrow{M},\overrightarrow{N}]}(\varphi_l^{(k)})=0,
\quad \left(\pa^{-1}T^{[\overrightarrow{M},\overrightarrow{N}]-1*}\pa^{k+1}\right)(\psi_p)=0,$$
then \eqref{LMNKnew1} holds.
\end{proof}
\subsection{Proof of Theorem 2}
Next let us proceed to \textbf{prove Theorem 2}. For this, let us set
\begin{align*}
  L^{[M+N]} &= T^{[\overrightarrow{M},\overrightarrow{N}]}\cdot \pa^k\cdot T^{[\overrightarrow{M},\overrightarrow{N}]-1}, \\
  q_i^{[M+N]} &=
\begin{cases}
-T^{[\overrightarrow{M},\overrightarrow{N}]} (f_i^{((M_i+1)k)}), & 1 \leq i \leq s, \\
\Biggl( \left(  \pa^{-1}  T^{[\overrightarrow{M},\overrightarrow{N}\setminus \{N[i-s]\}]-1*} \pa  \right) (g_{i-s}^{(N_{i-s}k)}) \Biggr)^{-1}, & s+1 \leq i \leq m,
\end{cases}\\
r_i^{[M+N]} &=
\begin{cases}
\biggl( T^{[\overrightarrow{M}\setminus \{M[i]\},\overrightarrow{N}]}(f_i^{(M_i k)})\biggr)^{-1}, & 1 \leq i \leq s, \\
(-1)^{k+1}\left(\pa^{-1}T^{[\overrightarrow{M},\overrightarrow{N}]-1*}\pa\right)(g_{i-s}^{((N_{i-s}+1) k)}), & s+1 \leq i \leq m.
\end{cases}
\end{align*}
Then by the determinant representations of $T^{[M,N]}$ in Appendix, we can know:\\
$\bullet ~ when ~ 1 \leq i \leq s, $
\begin{align*}
q_i^{[M+N]} &=
\begin{cases}
\displaystyle
-\frac{IW_{N,M+1}\big(\mathfrak{r}_{\overrightarrow{N},x}; \mathfrak{q}_{\overleftarrow{M}}, f_i^{((M_i+1)k)}\big)}
     {IW_{N,M+1}\big(\mathfrak{r}_{\overrightarrow{N},x}; \mathfrak{q}_{\overleftarrow{M}},1\big)}, & M\geq N, \\
\displaystyle
(-1)^{M+1}\frac{ IW_{M+1,N}\big(\mathfrak{q}_{\overrightarrow{M}}, f_i^{((M_i+1)k)}; \mathfrak{r}_{\overleftarrow{N},x}\big)}
     {IW_{M+1,N}\big(1,\mathfrak{q}_{\overrightarrow{M}};\mathfrak{r}_{\overleftarrow{N},x}\big)}, & M< N,
\end{cases}
\end{align*}
\begin{align*}
r_i^{[M+N]} &=
\begin{cases}
\displaystyle
(-1)^{M[i]+M}\frac{IW_{N,M}\big(\mathfrak{r}_{\overrightarrow{N},x}; \mathfrak{q}_{\overleftarrow{M}\setminus \{M[i]\}},1 \big)}
     {IW_{N,M}\big(\mathfrak{r}_{\overrightarrow{N},x}; \mathfrak{q}_{\overleftarrow{M}}\big)}, & M> N, \\
\displaystyle
(-1)^{M[i]+M}\frac{ IW_{M,N}\big(1,\mathfrak{q}_{\overrightarrow{M}\setminus \{M[i]\}}; \mathfrak{r}_{\overleftarrow{N},x}\big)}
     {IW_{M,N}\big(\mathfrak{q}_{\overrightarrow{M}};\mathfrak{r}_{\overleftarrow{N},x}\big)}, & M\leq N,
\end{cases}
\end{align*}
$\bullet ~ when ~ s+1 \leq i \leq m, $
\begin{align*}
q_i^{[M+N]} &=
\begin{cases}
\displaystyle
(-1)^{N[i-s]+N}\frac{ IW_{N,M}\big(1,\mathfrak{r}_{\overrightarrow{N}\setminus\{N[i-s]\}}; \mathfrak{q}_{\overleftarrow{M},x}\big)}
     {IW_{N,M}\big(\mathfrak{r}_{\overrightarrow{N}}; \mathfrak{q}_{\overleftarrow{M},x}\big)}, & M\geq N, \\
\displaystyle
(-1)^{N[i-s]+N}\frac{ IW_{M,N}\big(\mathfrak{q}_{\overrightarrow{M},x}; \mathfrak{r}_{\overleftarrow{N}\setminus\{N[i-s]\} },1\big)}
     {IW_{M,N}\big(\mathfrak{q}_{\overrightarrow{M},x};\mathfrak{r}_{\overleftarrow{N}}\big)}, & M< N,
\end{cases}
\end{align*}
\begin{align*}
r_i^{[M+N]} &=
\begin{cases}
\displaystyle
(-1)^{k+M}\frac{IW_{N+1,M}\big(\mathfrak{r}_{\overrightarrow{N}},g_{i-s}^{((N_{i-s}+1)k)}; \mathfrak{q}_{\overleftarrow{M},x}\big) }
     {IW_{N+1,M}\big(1,\mathfrak{r}_{\overrightarrow{N}}; \mathfrak{q}_{\overleftarrow{M},x}\big)}, & M> N, \\
\displaystyle
(-1)^{k+1}\frac{ IW_{M,N+1}\big(\mathfrak{q}_{\overrightarrow{M},x}; \mathfrak{r}_{\overleftarrow{N}},g_{i-s}^{((N_{i-s}+1)k)}\big)}
     {IW_{M,N+1}\big(\mathfrak{q}_{\overrightarrow{M},x};\mathfrak{r}_{\overleftarrow{N}},1\big)}, & M\leq N,
\end{cases}
\end{align*}

By Corollary 2, we have
\begin{align}\label{L[M+N]k}
(L^{[M+N]k})_{\leq 0} = \sum_{i=1}^m q_i^{[M+N]} \pa^{-1} r_i^{[M+N]} \pa.
\end{align}
Further since $L=\pa$ is the 1st mKP Lax operator, and $T^{[\overrightarrow{M},\overrightarrow{N}]}$ is the corresponding Darboux operator, therefore the following three equations hold,
\begin{align}
&(L^{[M+N]})_{t_n} = \left[ (L^{[M+N]n})_{\geq 1}, L^{[M+N]} \right], \\
&q_{i,t_n}^{[M+N]} = \left(L^{[M+N]n}\right)_{\geq 1} (q_i^{[M+N]}), \\
&r_{i,t_n}^{[M+N]} = -\left(\pa^{-1} (L^{[M+N]n})_{\geq 1}^* \pa \right)( r_i^{[M+N]} ).\label{ri,tn[M+N]}
\end{align}
\eqref{L[M+N]k}-\eqref{ri,tn[M+N]} are precisely the complete characterization of the $(k,m)$-constrained mKP hierarchy. For $L=\pa$ is the 1st mKP Lax operator, its corresponding tau functions are $\tau_0=\tau_1=1.$ So by \hspace{-0.1pt}{\cite[Proposition 5.2]{cheng2018}} and Corollary 1, we know
the tau pair $(\tau_0^{[M+N]},\tau_1^{[M+N]})$ is the tau function corresponding to $L^{[M+N]}$, where $\tau_0^{[M+N]}$ and
$ \tau_1^{[M+N]}$ are given by
\begin{align*}
   \bullet ~ when ~ M \geq N, \nonumber \\
   \tau_0^{[M+N]} =& IW_{N,M}(\mathfrak{r}_{\overrightarrow{N},x};\ \mathfrak{q}_{\overleftarrow{M}}) ,\
    \tau_1^{[M+N]} = IW_{N,M+1}(\mathfrak{r}_{\overrightarrow{N},x};\ \mathfrak{q}_{\overleftarrow{M}},1), \\
  \bullet ~ when ~ M < N, \nonumber \\
   \tau_0^{[M+N]} =& IW_{M,N}(\mathfrak{q}_{\overrightarrow{M}};\ \mathfrak{r}_{\overleftarrow{N},x}),\
   \tau_1^{[M+N]} = (-1)^MIW_{M+1,N}(1,\mathfrak{q}_{\overrightarrow{M}};\ \mathfrak{r}_{\overleftarrow{N},x}).
  \end{align*}
Next set
$\tau_{2,i}^{[M+N]} = q_i^{[M+N]} \tau_1^{[M+N]}, \ \tau_{-1,i}^{[M+N]} = r_i^{[M+N]} \tau_0^{[M+N]}$, we have the expressions of $\tau_{2,i}^{[M+N]}$ and $\tau_{-1,i}^{[M+N]}$ given in Theorem 2.
Next by \hspace{-0.1pt}{\cite[Proposition 8]{wu2022MMMA}}, we can find $\tau_0^{[M+N]}$,\ $\tau_1^{[M+N]}$,\ $\tau_{2,i}^{[M+N]}$ and $\tau_{-1,i}^{[M+N]}$ satisfy the $(k,m)$-constrained mKP bilinear equations \eqref{kmcmkp1}~-~\eqref{kmcmkp3}. Thus, we complete $\textbf{the proof of Theorem 2}$.

\section{Examples of (k,m)-constrained mKP hierarchy}\label{section4}
For the $(k,m)$-constrained mKP hierarchy \eqref{mKPLK<1} and \eqref{mKPeigennew}, let us assume
\begin{align}\label{LKall}
L^k=\pa^k + \sum_{i=1}^{k-1} V_i \pa^i + \sum_{i=1}^m q_i \pa^{-1} r_i \pa,
\end{align}
where $V_i$ $(1\leq i\leq k-1)$ are differential polynomials of $v_j$  $(j\geq 0)$.
\subsection{The (1,m)-constrained mKP hierarchy}
When $k=1$, \eqref{LKall} will become into
\begin{align}\label{(1m)cmKP}
L=\pa+\sum_{i=1}^{m}q_i\pa^{-1}r_i\pa.
\end{align}
Compared with the mKP Lax operator $L=\pa+\sum_{i=0}^{+\infty}v_i\pa^{-i}$,\ we can get
\begin{align*}
v_0=\sum_{i=1}^{m}q_ir_i,\quad v_1=-\sum_{i=1}^{m}q_ir_{i,x},\quad v_2=\sum_{i=1}^{m}q_ir_{i,xx},\ \ldots.
\end{align*}
Then by \eqref{Lax equation},\ we can obtain the following examples.
\begin{example}
$q_i$ and $r_i$ satisfy the following equations $(1\leq i \leq m)$
\begin{align*}
q_{i,t_2}=q_{i,xx}+2q_{i,x}\sum_{j=1}^{m}q_jr_j,\quad
r_{i,t_2}=-r_{i,xx}+2r_{i,x}\sum_{j=1}^{m}q_jr_j.
\end{align*}
In particular, when $m=1$, we denote $q\triangleq q_1$ and $r\triangleq r_1$, then
\begin{align}\label{qt2rt2}
q_{t_2}=q_{xx}+2qq_xr,\quad r_{t_2}=-r_{xx}+2qrr_x.
\end{align}

If we choose $f=e^{x+t_2}+e^{2x+4t_2}+e^{3x+9t_2}$, and take $K=0,s=1,M_1=1,a=0,Q=0$ in \textbf{Theorem 2}, then $M=2$ and $N=0$, and we have
\begin{align*}
&\tau_0^{[2]} =
\begin{vmatrix}
f & f_x \\
f_x & f_{xx}
\end{vmatrix}=e^{3x+5t_2} + 4e^{4x+10t_2} + e^{5x+13t_2},\\
&\tau_1^{[2]} =
\begin{vmatrix}
f & f_x & 1 \\
f_x & f_{xx} & 0 \\
f_{xx} & f_{xxx} & 0
\end{vmatrix}= 2e^{3x+5t_2} + 12e^{4x+10t_2} + 6e^{5x+13t_2},\\
&\tau_{-1}^{[2]} =-
\begin{vmatrix}
f & 1 \\
f_x & 0
\end{vmatrix}=e^{x+t_2}+2e^{2x+4t_2}+3e^{3x+9t_2},\\
&\tau_{2}^{[2]} =
\begin{vmatrix}
f & f_x & f_{xx} \\
f_x & f_{xx} & f_{xxx} \\
f_{xx} & f_{xxx} & f_{xxxx}
\end{vmatrix}=4e^{6x+14t_2}.
\end{align*}
So by \eqref{tau2i-1i}, we can obtain the solutions of \eqref{qt2rt2},
\begin{align*}
q=&\frac{\tau_2^{[2]}}{\tau_1^{[2]}}=\frac{2e^{3x+9t_2}}{1+6e^{x+5t_2}+3e^{2x+8t_2} }, \\
r=&\frac{\tau_{-1}^{[2]}}{\tau_0^{[2]}}=\frac{e^{-2x-4t_2}(1+2e^{x+3t_2}+3e^{2x+8t_2})}{1+4e^{x+5t_2}+e^{2x+8t_2} }.
\end{align*}
\end{example}

\begin{example}
$q_i$ and $r_i$ also satisfy the following equations $(1\leq i \leq m)$
\begin{align*}
&q_{i,t_3}=q_{i,xxx} + 3q_{i,xx}\sum_{j=1}^m q_j r_j  + 3q_{i,x}\sum_{j=1}^m q_{j,x} r_j  + 3q_{i,x}\left(\sum_{j=1}^m q_j r_j\right)^2 ,\\
&r_{i,t_3}=r_{i,xxx} - 3r_{i,xx}\sum_{j=1}^m q_j r_j  - 3r_{i,x}\sum_{j=1}^m q_j r_{j,x}  + 3r_{i,x}\left(\sum_{j=1}^m q_j r_j\right)^2.
\end{align*}
In particular, when $m=1$,
\begin{align}
q_{t_3}&=q_{xxx} + 3qq_{xx}r + 3q_{x}^2r + 3q_{x}(qr)^2 ,\label{qt3}\\
r_{t_3}&=r_{xxx} - 3qrr_{xx} - 3qr_{x}^2  + 3r_{x}(qr)^2.\label{rt3}
\end{align}
Similarly if we choose $f=\frac{1}{24}x^4+xt_3,\ \psi=e^{-2x-8t_3}$, and take $K=0, s=1, M_1=2, Q=1, a=0, N_1=0$ in \textbf{Theorem 2}, then $M=3,\ N=1$, and we have
\begin{align*}
&\tau_0^{[4]} =- \frac{e^{-2x - 8t_3} (x^3 - 12 t_3) (x^3 + 3 x^2 + 3 x- 12 t_3)}{576},\\
&\tau_1^{[4]} = -\frac{e^{-2x - 8t_3} (2x^3 + 3x^2 + 12t_3)}{96},\\
&\tau_{-1}^{[4]} =\frac{e^{-2x - 8t_3} (2x^4+4x^3+3x^2-24xt_3-12t_3)}{96},\\
&\tau_{2}^{[4]} =- \frac{e^{-2x - 8t_3} (x^4 + 2x^3 + 24xt_3 + 12t_3)}{192}.
\end{align*}
So by \eqref{tau2i-1i}, we can obtain the solutions of \eqref{qt3} and \eqref{rt3},
\begin{align*}
  q=&\frac{\tau_2^{[4]}}{\tau_1^{[4]}}=\frac{x^4+2x^3+24xt_3+12t_3}{2(2x^3+3x^2+12t_3)}, \\
  r=&\frac{\tau_{-1}^{[4]}}{\tau_0^{[4]}}=\frac{-6(2x^4+4x^3+3x^2-24xt_3-12t_3)}{(x^3-12t_3)(x^3+3x^2+3x-12t_3)}.
\end{align*}
\end{example}

\subsection{The (2,2)-constrained mKP hierarchy}
When $k=2,\ m=2$, the mKP Lax operator satisfies
\begin{align}
L^2=\pa^2 + V \pa +  q_1 \pa^{-1} r_1 \pa+q_2 \pa^{-1} r_2 \pa.
\end{align}
So we have
\begin{align*}
v_0=&\frac{1}{2}V,\quad v_1=\frac{1}{2}(q_1r_1+q_2r_2-\frac{1}{4}V^2-\frac{1}{2}V_{x}),\\
v_2 =& -\frac{3}{4} q_1r_{1,x}-q_2r_{2,x}-\frac{1}{4}q_{1,x}r_1-\frac{1}{4}q_{2,x}r_2 + \frac{1}{8}V_{xx}\\
&+\frac{1}{4}VV_x -\frac{1}{4}V(q_1r_1+q_2r_2-\frac{1}{4}V^2),\ldots,
\end{align*}
and in terms of tau functions $\tau_i(-1\leq i\leq 2)$, we have
\begin{align*}
&q_1(t)=\frac{\tau_{2,1}(t)}{\tau_1(t)},\quad r_1(t)=\frac{\tau_{-1,1}(t)}{\tau_0(t)},\\
&q_2(t)=\frac{\tau_{2,2}(t)}{\tau_1(t)},\quad r_2(t)=\frac{\tau_{-1,2}(t)}{\tau_0(t)}, \\ &V=2\left(\log\biggl(\frac{\tau_1(t)}{\tau_0(t)}\biggr)\right)_x.
\end{align*}
Then we can obtain the following examples.
\begin{example}
$q_i,r_i$ and $V$ satisfy the following equations $(1\leq i \leq m)$
\begin{align}
q_{1,t_2}=&q_{1,xx}+Vq_{1,x},\quad
r_{1,t_2}=-r_{1,xx}+Vr_{1,x},\label{q1r1t2}\\
q_{2,t_2}=&q_{2,xx}+Vq_{2,x},\quad
r_{2,t_2}=-r_{2,xx}+Vr_{2,x},\label{q2r2t2}\\
V_{t_2}=&2q_1r_{1,x}+2q_2r_{2,x}+2q_{1,x}r_1+2q_{2,x}r_2;\label{Vt2}
\end{align}
Similarly if we choose $f=e^{x+t_2}+e^{2x+4t_2}$ and $g=e^{-3x-9t_2}+e^{-4x-16t_2}$, and set $s=1, K=M_1=0, a=1, Q=N_1=0$ in \textbf{Theorem 2}, then $M=N=1$, and we can get
\begin{align*}
&\tau_0^{[2]} =\int fg_xdx=\frac{3}{2} e^{-2x - 8t_2} + 3 e^{-x - 5t_2} + \frac{4}{3} e^{-3x - 15t_2} + 2 e^{-2x - 12t_2},\\
&\tau_1^{[2]} =
\begin{vmatrix}
\int fg_x dx & g \\
f & 1
\end{vmatrix}=\frac{1}{6} e^{-3x-15t_2}( 3 e^{x+7t_2} + 12 e^{2x+10t_2} + 6 e^{x+3t_2}+2 ),\\
&\tau_{-1,1}^{[2]} =g=e^{-3x-9t_2}+e^{-4x-16t_2},\\
&\tau_{2,1}^{[2]} =-
\begin{vmatrix}
\int fg_x dx & \int f_{xx}g_x dx \\
f & f_{xx}
\end{vmatrix}= \frac{1}{2}e^{-3x-15t_2} ( 9 e^{3x+11t_2} + 4 e^{2x+4t_2} ),\\
&\tau_{-1,2}^{[2]} =-\frac{\begin{vmatrix}
\int f_xg dx & \int f_xg_{xx} dx \\
g & g_{xx}
\end{vmatrix}\cdot\int fg_x dx}{
\begin{vmatrix}
\int f_xg dx &  f \\
g & 1
\end{vmatrix}}\\
&\quad  \quad=-\frac{7}{6} \frac{\left(8 e^{-3 x-15 t_2}+9 e^{-2 x-8 t_2}+18 e^{-x-5 t_2}+12 e^{-2 x-12 t_2}\right)\left(6 e^{-2 x-6 t_2}+e^{-3 x-9 t_2}\right)}{8+12 e^{x+3 t_2}+9 e^{x+7 t_2}+18 e^{2 x+10 t_2}},\\
&\tau_{2,2}^{[2]} =
\frac{f\cdot \begin{vmatrix}
\int fg_x dx & g \\
f & 1
\end{vmatrix}}{\int f_xgdx}
=-e^{x+t_2}-e^{2x+4t_2}.
\end{align*}
So by \eqref{tau2i-1i}, we can obtain the solutions of \eqref{q1r1t2}-\eqref{Vt2},
\begin{align*}
  q_1&=\frac{\tau_{2,1}^{[2]}}{\tau_1^{[2]}}=\frac{3e^{-x-5t_2}(48e^{-3x-15t_2}+16e^{-4x-22t_2}+27e^{-2x-8t_2})}
{36e^{-4x-14t_2}+66e^{-5x-21t_2}+9e^{-5x-17t_2}+18e^{-6x-24t_2}+24e^{-6x-28t_2}+8e^{-7x-31t_2}},\\
r_1&=\frac{\tau_{-1,1}^{[2]}}{\tau_0^{[2]}}=-\frac{7(3e^{-5x-17t_2}+18e^{-4x-14t_2}+4e^{-6x-24t_2}+24e^{-5x-21t_2})}
{(3e^{x+7t_2}+4)(8e^{-3x-15t_2}+9e^{-2x-8t_2}+18e^{-x-5t_2}+12e^{-2x-12t_2})},\\
q_2&=\frac{\tau_{2,2}^{[2]}}{\tau_1^{[2]}}=-\frac{6(4e^{-3x-15t_2}+3e^{-x-5t_2}+3e^{-2x-8t_2}+4e^{-2x-12t_2})}
{36e^{-4x-14t_2}+66e^{-5x-21t_2}+9e^{-5x-17t_2}+18e^{-6x-24t_2}+24e^{-6x-28t_2}+8e^{-7x-31t_2}}, \\
r_2&=\frac{\tau_{-1,2}^{[2]}}{\tau_0^{[2]}}=\frac{6(e^{-3x-9t_2}+e^{-4x-16t_2})}
{8e^{-3x-15t_2}+9e^{-2x-8t_2}+18e^{-x-5t_2}+12e^{-2x-12t_2}},\\
V&=\frac{A(x,t_2)}{B(x,t_2)},
\end{align*}
where $A(x,t_2)=12( 81 e^{-9x-31t_2}+144 e^{-10x-42t_2}+96 e^{-11x-49t_2}+54 e^{-9x-35t_2}+16 e^{-13x-55t_2}+9 e^{-11x-41t_2}+396 e^{-10x-38t_2}+64 e^{-13x-59t_2}+416 e^{-12x-52t_2}+660 e^{-11x-45t_2}+24 e^{-12x-48t_2} )$ and $B(x,t_2)=( 8 e^{-3x-15t_2}+9 e^{-2x-8t_2}+18 e^{-x-5t_2}+12 e^{-2x-12t_2} )
( 3 e^{-3x-9t_2}+4 e^{-4x-16t_2} )( 36 e^{-4x-14t_2}+66 e^{-5x-21t_2}+9 e^{-5x-17t_2}+18 e^{-6x-24t_2}+24 e^{-6x-28t_2}+8 e^{-7x-31t_2} )$.
\end{example}
\section{Conclusions and discussions}\label{section5}
In this paper, the generalized Wronskian solutions of the $(k,m)$-constrained mKP hierarchy are investigated. We first give an elementary transformation of the mKP tau pair in Theorem \ref{th1}, then we give more general forms of the mKP tau pairs with generalized Wronskians using the mKP Darboux transformations, as given in Corollary \ref{corollary1}. Finally starting from $L=\partial$, the generalized Wronskian solutions for the $(k,m)$-constrained mKP hierarchy are obtained to satisfy the corresponding bilinear equations \eqref{kmcmkp1}-\eqref{kmcmkp3} in Theorem \ref{th2}. Here are the following remarks on our results.
\begin{itemize}
\item Corollary \ref{corollary1} is related to the addition formulas of the mKP hierarchy.  In fact for the mKP wave function $w(t,z)$ and the mKP adjoint wave function $w^*(t,z)$, one has the following two important relations
\begin{align}
w(t,z)_x= z \frac{\tau_0(t)\tau_1(t-[z^{-1}])}{\tau_1^2(t)} e^{\xi(t,z)},\quad w^{*}(t,z)_x= -\frac{\tau_0(t+[z^{-1}])\tau_1(t)}{\tau_0^2(t)} e^{-\xi(t,z)}.
\end{align}
which can be found in \cite{cheng2018kp}. So if we set $q(t)=w(t,z)$ and $r(t)=w^*(t,z)$ in Theorem \ref{th1}, then we can find
\begin{align*}
X(t,z)\tau_0(t)&=w(t,z)\tau_1(t),\quad X(t,z)\tau_1(t)=z^{-1}w(t,z)_x\frac{\tau_1(t)^2}{\tau_0(t)},\\
X^*(t,z)\tau_0(t)&=-w^{*}(t,z)_x\frac{\tau_0(t)^2}{\tau_1(t)},\quad X^*(t,z)\tau_1(t)=zw^*(t,z)\tau_0(t),
\end{align*}
where $X(t,z)$ and $X^*(t,z)$ are the vertex operators defined by
$$X(t,z)=e^{\xi(t,z)}e^{-\xi(\widetilde{\partial}_t,z^{-1})},\quad X^*(t,z)=e^{-\xi(t,z)}e^{\xi(\widetilde{\partial}_t,z^{-1})},$$
And $(X(t,z)\tau_0(t), X(t,z)\tau_1(t))$ and $(X^*(t,z)\tau_0(t), X^*(t,z)\tau_1(t))$ are the new mKP tau pairs. 
Further if we set $q_i(t)=w(t,\lambda_i)$ ($1\leq i\leq M$) and $r_j(t)=w^*(t,\mu_j)$ ($1\leq j\leq N$) in Corollary \ref{corollary1}, then we can obtain the addition formulas of the 1st mKP hierarchy in the forms of determinant \cite{Shigyo2013}.

\item Though in \cite{Chen2020,Yangyi2020}, the $(k,m)$-constrained mKP hierarchy is investigated by the Darboux transformations, however there is no explicit solution, since they highly rely on the choice of the initial mKP eigenfunctions and adjoint eigenfunctions. Notice that when $m=1$, The initial Lax operator can be choose as $L^k=\partial^k+\partial^{-1}$, and it is usually very difficult to solve $q_t=(L^n)_{\geq 1}(q)$ and $r_t=-(\pa^{-1} (L^n)_{\geq 1}^* \pa)(r)$. Moreover for $m>1$, it will become more complicated to get the initial mKP eigenfunctions and the adjoint eigenfunctions.

\item If we set $f_i(t)$ and $g_j(t)$ to be some Schur functions $s_\lambda (t)={\rm det} (h_{\lambda_i-i+j}(t))_{1\leq i,j\leq l(\lambda)},\ e^{\xi(t,z)}=\sum_{j=0}^{+\infty}h_j(t)z^j$ in Theorem \ref{th2}, we can obtain the polynomial solutions for the constrained mKP hierarchy by the formula below \cite{Orlov2002},
$$\langle s - k | e^{H(t)} \psi^*_{-j_1} \cdots \psi^*_{-j_k} \psi_{i_s} \cdots \psi_{i_1} | 0 \rangle = (-1)^{j_1 + \cdots + j_k + (k-s)(k-s+1)/2} s_\lambda(t), $$
here $-j_1 < \cdots < -j_k < 0 \le i_s < \cdots < i_1$,\ $s - k \ge 0$, and the partition $\lambda = (n_1, \ldots, n_{s-k}, n_{s-k+1},$\\ $\ldots, n_{s-k+j_1})$ is defined by the pair of partitions:
\begin{align*}
&(n_1, \ldots, n_{s-k}) = \bigl( i_1 - (s - k) + 1,\, i_2 - (s - k) + 2,\, \ldots,\, i_{s-k} \bigr), \\
&(n_{s-k+1}, \ldots, n_{s-k+j_1}) = (i_{s-k+1}, \ldots, i_s \,|\, j_1 - 1, \ldots, j_k - 1). 
\end{align*}
\item By comparing the results of $q_i$ and $r_i$ $(1\leq i\leq m)$ in Theorem \ref{th2}, if we need $q_i=\bar{r}_j$ with $\bar{r}_j$ meaning the complex conjugate of $r_j$, which is needed in the vector derivative NLS equation, it will be interesting to investigate the additional conditions for $f_i$ and $g_j$, which will be considered in future.
\end{itemize}

\section{Appendix}
In this appendix, we will give the expressions of $T^{[\overrightarrow{M},\overrightarrow{N}]}$ and $\pa^{-1}{T}^{[\overrightarrow{M},\overrightarrow{N}]-1*} \pa $, where $T^{[\overrightarrow{M},\overrightarrow{N}]}=T^{[\overrightarrow{M},\overrightarrow{N}]}(f_{\overleftarrow{M}};g_{\overleftarrow{N},x})$ is given by \eqref{Tnkexpr}. Firstly the expressions of $T^{[\overrightarrow{M},\overrightarrow{N}]}(f_{\overleftarrow{M}};g_{\overleftarrow{N},x})$ are given by

$\bullet$ Case $M>N$,
\begin{align*}
&\frac{1}{I W_{N,M+1}(g_{\overrightarrow{N},x}; f_{\overleftarrow{M}},1)}\begin{vmatrix}
\Omega(f_1, g_{N,x}) & \cdots & \Omega(f_M, g_{N,x}) & \partial^{-1}g_{N,x} \\
\vdots & \ddots & \vdots & \vdots \\
\Omega(f_1, g_{1,x}) & \cdots & \Omega(f_M, g_{1,x}) & \partial^{-1}g_{1,x} \\
f_1 & \cdots & f_M & 1 \\
f_1^{(1)} & \cdots & f_M^{(1)} & \partial \\
\vdots & \ddots & \vdots & \vdots \\
f_1^{(M-N)} & \cdots & f_M^{(M-N)} & \partial^{M-N}
\end{vmatrix},
\end{align*}

$\bullet$ Case $M=N$,
\begin{align*}
 &\frac{1}{{I W_{M,M+1}(g_{\overrightarrow{M},x}; f_{\overleftarrow{M}},1)}}
\begin{vmatrix}
\Omega(f_1, g_{M,x} ) & \cdots & \Omega(f_M, g_{M,x}) & \partial^{-1} g_{M,x} \\
\vdots & \ddots & \vdots & \vdots \\
\Omega(f_1, g_{1,x}) & \cdots & \Omega(f_M, g_{1,x}) & \partial^{-1} g_{1,x} \\
f_1 & \cdots & f_M & 1
\end{vmatrix},
\end{align*}

$\bullet$ Case $M<N$,
\begin{align*}
\frac{(-1)^{N-1}}{IW_{M+1,N}(1, f_{\overrightarrow{M}}; g_{\overleftarrow{N},x})}
\begin{vmatrix}
 \Omega(f_M, g_{1,x})& \cdots & \Omega(f_1, g_{1,x}) & g_{1,x} & \cdots & g_1^{(N-M-1)} & \partial^{-1} g_{1,x} \\
 \Omega(f_M, g_{2,x}) & \cdots & \Omega(f_1, g_{2,x}) & g_{2,x} & \cdots & g_2^{(N-M-1)} & \partial^{-1} g_{2,x} \\
 \vdots & \ddots & \vdots & \vdots & & \vdots & \vdots \\
 \Omega(f_M, g_{N,x}) & \cdots & \Omega(f_1, g_{N,x}) & g_{N,x} & \cdots & g_N^{(N-M-1)} & \partial^{-1} g_{N,x} \\
\end{vmatrix},
\end{align*}
Then as for $\pa^{-1}{T}^{[\overrightarrow{M},\overrightarrow{N}]}(f_{\overleftarrow{M}};g_{\overleftarrow{N},x})^{-1*} \pa $, we need the following key relations between $T_D(f)$ and $T_I(g)$, that is,
\begin{align*}
\pa^{-1}T_D(f)^{-1*}\pa=T_I(f),\quad \pa^{-1}T_I(g)^{-1*}\pa=T_D(g).
\end{align*}
Thus we can find 
$$\pa^{-1}{T}^{[\overrightarrow{M},\overrightarrow{N}]}(f_{\overleftarrow{M}};g_{\overleftarrow{N},x})^{-1*} \pa=T^{[\overrightarrow{N},\overrightarrow{M}]}(g_{\overleftarrow{N},x};f_{\overleftarrow{M}}).$$
So we can easily obtain the corresponding determinant representations.

\noindent{\bf Acknowledgements}: \\
This work is supported by the National Natural Science Foundation of China
(Grant Nos. 12571271 and 12261072).\\

\noindent{\bf Conflict of Interest}: \\
The author has no conflicts to disclose.\\

\noindent{\bf Data availability}: \\
Data sharing is not applicable to this article as no new data were created or analyzed in this study.\\


\begin{thebibliography}{99}
\bibitem{Adler1999} Adler M and van Moerbeke P, Vertex operator solutions to the discrete KP hierarchy, Commun. Math. Phys. 203 (1999) 185-210.
\bibitem{Alexandrov2015} Alexandrov A, Open intersection numbers, matrix models and MKP hierarchy, J. High Energy Phys. 3 (2015) 1-14.
\bibitem{Cao2026} Cao W Q, Cheng J P and Wang J B, The two-component discrete KP hierarchy, Stud. Appl. Math. 156 (2026) 70213.
\bibitem{Chen20193} Chen H Z, Geng L M, Li N and Cheng J P, Solving the constrained modified KP hierarchy by gauge transformations, J. Nonlin. Math. Phys. 26 (2019) 54-68.
\bibitem{Chen2020} Chen H Z, Geng L M and Cheng J P, Solutions of the constrained mKP hierarchy by boson-fermion correspondence, J. Nonlin. Math. Phys. 27 (2020) 308-323.
\bibitem{cheng2018} Cheng J P, The gauge transformation of the modified KP hierarchy, J. Nonlin. Math. Phys. 25 (2018) 66-85.
\bibitem{cheng2018kp}
Cheng J P, Li M H and Tian K L, On the modified KP hierarchy: tau functions, squared eigenfunction symmetries and additional symmetries. J. Geom. Phys. 134 (2018) 19-37.
\bibitem{DJKM} Date E, Kashiwara M, Jimbo M and Miwa T, Transformation groups for soliton equations, in Nonlinear integrable systems classical theory and quantum theory (Kyoto, 1981), 39-119, World Sci. Publ, Singapore, 1983.
\bibitem{Dickey1999} Dickey L A, Modified KP and discrete KP, Lett. Math. Phys. 48 (1999) 277-289.
\bibitem{Dickey2003} Dickey L A, Soliton Equations and Hamiltonian Systems, World Sci. Publ, Singapore,  2003.
\bibitem{Harnad2021} Harnad J and Balogh F, Tau functions and their applications, Cambridge Monogr. Math. Phys., Cambridge Univ. Press, Cambridge, 2021.
\bibitem{He2002} He J S, Li Y S and Cheng Y, The determinant representation of the gauge transformation operators, Chinese Ann. Math. Ser. B 23 (2002) 475-486.
\bibitem{Jimbo} Jimbo M and Miwa T, Solitons and infinite dimensional Lie algebras, Publ. RIMS, Kyoto Univ. 19 (1983) 943-1001.
\bibitem{Kac2018} Kac V and van de Leur J, Equivalence of formulations of the MKP hierarchy and its polynomial tau-functions, Jpn. J. Math. 13 (2018) 235-271.
\bibitem{Kac2023} Kac V and van de Leur J, Multicomponent KP type hierarchies and their reductions, associated to conjugacy classes of Weyl groups of classical Lie algebras, J. Math. Phys. 64 (2023) 091702.
\bibitem{kiso1990} Kiso K, A remark on the commuting flows defined by Lax equations, Progr. Theoret. Phys. 83 (1990) 1108-1114.
\bibitem{Kun1995} Kundu A, Strampp W and Oevel W, Gauge transformations of constrained KP flows: new integrable hierarchies, J. Math. Phys. 36 (1995) 2972-2984.
\bibitem{kupershmidt1985}
Kupershmidt B A, Mathematics of dispersive water waves, Commun. Math. Phys. 99 (1985) 51-73.
\bibitem{Liu1995} Liu Q P, On the constrained modified KP hierarchy, Inverse Probl. 11 (1995) 205-209.
\bibitem{Mulase1994} Mulase M, Algebraic theory of the KP equations, Conf. Proc. Lecture Notes Math. Phys. 3 (1994) 151-217.
\bibitem{Oevel1999JP} Oevel W and Carillo S, Squared eigenfunction symmetries for soliton equations: Part II, J. Math. Anal. Appl. 217 (1998) 179-199.
\bibitem{Oevel1998JP} Oevel W and Rogers C, Gauge transformations and reciprocal links in 2+1 dimensions, Rev. Math. Phys. 5 (1993) 299-330.
\bibitem{Orlov2002} Orlov A Y, Tau Functions and Matrix Integrals, arXiv: math-ph/0210012v3, 2002.
\bibitem{Zabrodin2026} Prokofev V and Zabrodin A, Bilinear formalism for Schwarzian KP and Harry Dym hierarchies, arXiv:2605.00235.
\bibitem{Shaw1997} Shaw J C and Tu M H, Miura and auto-B$\ddot{a}$cklund transformations for the cKP and cmKP hierarchies, J. Math. Phys. 38 (1997) 5756-5773.
\bibitem{Shigyo2013} Shigyo Y, On addition formulae of KP, mKP and BKP hierarchies, SIGMA 9 (2013) 035.
\bibitem{vanMoerbeke1994} van Moerbeke P, Integrable foundations of string theory, Lectures on integrable systems, World Sci. Publ., River Edge, NJ, 1994, 163-267.
\bibitem{wang2026} Wang J B, Guan W C, Chen M Y and Cheng J P, One reduction of the modified Toda hierarchy, Anal. Math. Phys.  16 (2026) 50.
\bibitem{Willox2004} Willox R and Satsuma J, Sato theory and transformation groups. A unified approach to integrable systems, Lect. Notes Phys. 644 (2004) 17-55.
\bibitem{wu2022MMMA} Wu Y Q and Cheng J P, A new generalized constrained modified KP hierarchy, Math. Meth. Appl. Sci. 46 (2023) 3510-3521.
\bibitem{Yangyi2020} Yang Y and Cheng J P, The gauge transformations generated by the wave functions in the constrained modified KP hierarchy, Modern Phys. Lett. B 34 (2020) 2050205.
\bibitem{Yang2022} Yang Y and Cheng J P, Bilinear equations in Darboux transformations by Boson-Fermion correspondence, Phys. D 433 (2022)  133198.
\bibitem{Zabrodin2019} Zabrodin A V, Matrix Modified Kadomtsev-Petviashvili Hierarchy, Theor. Math. Phys. 199 (2019) 771-783.
\bibitem{Zabrodin225} Zabrodin A V, Revisiting B$\ddot{a}$cklund-Darboux transformations for KP and BKP integrable hierarchies, arXiv:2506.07208v1.
\end{thebibliography}
\end{document}